\documentclass[%
reprint,
superscriptaddress,
amsmath,
amssymb,
aps,
pra,
]{revtex4-2}

\usepackage{graphicx}
\usepackage{svg}
\usepackage{bm}
\usepackage{hyperref}
\usepackage{cleveref}
\usepackage{xcolor}
\usepackage{physics}
\usepackage{algpseudocodex}
\usepackage{amsthm}
\usepackage{csquotes}

\newtheorem{theorem}{Theorem}
\newtheorem{claim}{Claim}
\newtheorem{cor}{Corollary}

\newtheorem{remark}{Remark}
\newtheorem{definition}{Definition}

\crefname{claim}{Claim}{Claims}
\crefname{cor}{Corollary}{Corollaries}

\newcommand{\LDSset}[0]{\mathcal{S}_{\text{LDS}}} 
\newcommand{\probapoly}[1]{\mathcal{P}^{#1}} 
\newcommand{\localpoly}[1]{\mathcal{L}^{#1}} 
\newcommand{\quantumset}[1]{\mathcal{Q}^{#1}} 
\newcommand{\nspoly}[1]{\mathcal{NS}^{#1}} 
\newcommand{\simpcycleSet}[1]{\mathcal{C}_{#1}} 
\newcommand{\err}[3]{{\operatorname{err}_{#1}(#2,#3)}} 
\newcommand{\errtilde}[3]{\widetilde{\operatorname{err}}_{#1}(#2,#3)} 
\newcommand{\partySubSet}{\mathcal{I}}

\newcommand{\corrvecNoOne}[1]{\expval*{\vec{A}{#1}}}
\newcommand{\corrvecNoOneNBody}[1]{\expval*{\vec{\boldsymbol{A}}{#1}}}

\newcommand{\TINNproj}{{{\Pi_1}}}
\newcommand{\TINNNproj}{{\Pi_2}}
\newcommand{\TIRproj}{{\Pi_R}}

\newcommand{\Conv}[0]{\mathrm{Conv}} 
\newcommand{\db}{\operatorname{DB}} 
\newcommand{\lcm}[1]{\operatorname{lcm}(#1)} 
\newcommand{\Break}{\State \textbf{break} } 

\begin{document}

\title{Characterizing Translation-Invariant Bell Inequalities using Tropical Algebra and Graph Polytopes}

\author{Mengyao Hu}
\thanks{These two authors contributed equally.}
\affiliation{Instituut-Lorentz, Universiteit Leiden, P.O. Box 9506, 2300 RA Leiden, The Netherlands}
\author{Eloïc Vallée}
\thanks{These two authors contributed equally.}
\affiliation{Instituut-Lorentz, Universiteit Leiden, P.O. Box 9506, 2300 RA Leiden, The Netherlands}
\author{Tim Seynnaeve}
\affiliation{Departments of Computer Science and Mathematics, KU Leuven, Celestijnenlaan 200A,
3001 Leuven, Belgium}
\author{Patrick Emonts}
\affiliation{Instituut-Lorentz, Universiteit Leiden, P.O. Box 9506, 2300 RA Leiden, The Netherlands}
\author{Fatemeh Mohammadi}
\affiliation{Departments of Computer Science and Mathematics, KU Leuven, Celestijnenlaan 200A,
3001 Leuven, Belgium}
\author{Jordi Tura}
\affiliation{Instituut-Lorentz, Universiteit Leiden, P.O. Box 9506, 2300 RA Leiden, The Netherlands}

\date{\today}

\begin{abstract}
Nonlocality is one of the key features of quantum physics, which is revealed through the violation of a Bell inequality.
In large multipartite systems, nonlocality characterization quickly becomes a challenging task.
A common practice is to make use of symmetries, low-order correlators, or exploiting local geometries, to restrict the class of inequalities.
In this paper, we characterize translation-invariant (TI) Bell inequalities with finite-range correlators in one-dimensional geometries. 
We introduce a novel methodology based on tropical algebra tensor networks and highlight its connection to graph theory.
Surprisingly, we find that the TI Bell polytope has a number of extremal points that can be uniformly upper-bounded with respect to the system size.
We give an efficient method to list all vertices of the polytope for a particular system size, and characterize the tightness of a given TI Bell inequality.
The connections highlighted in our work allow us to re-interpret concepts developed in the fields of tropical algebra and graph theory in the context of Bell nonlocality, and vice-versa.
This work extends a parallel article [M. Hu \textit{et al.}, arXiv: 2208.02798 (2022)] on the same subject.

\end{abstract}

\maketitle

\section{Introduction}

Measurements on space-like separated quantum particles may lead to statistics that evade any classical explanation, even if assisted by shared randomness~\cite{brunner_bell_2014}. 
This phenomenon is known as nonlocality and is detected through the violation of a so-called Bell inequality~\cite{bell_einstein_1964,clauser_proposed_1969}.
In addition to being one of the intrinsic features of quantum physics at a foundational level, nonlocality is the key resource for device-independent (DI) quantum information processing tasks, such as DI quantum key distribution~\cite{pironio_device-independent_2009}, DI randomness expansion and amplification~\cite{pironio_random_2010} or DI self-testing~\cite{supic_self-testing_2020}. 
On the experimental side, tremendous progress has been achieved in recent years in performing so-called loophole-free Bell tests in bipartite systems~\cite{hensen_loophole-free_2015, giustina_significant-loophole-free_2015, shalm_strong_2015, rosenfeld_event-ready_2017,storz_loophole-free_2023}.

In multipartite systems, however, the role of nonlocality remains less explored, hindered by the intrinsic mathematical complexity associated to the problem and the much more demanding experimental requirements. 
Progress has relied on imposing extra mathematical structure, such as symmetries in the Bell inequalities~\cite{tura_translationally_2014,wang_entanglement_2017}, low-order correlators~\cite{tura_detecting_2014,tura_nonlocality_2015}, or designing inequalities tailored to a particular class of quantum states~\cite{collins_bell_2002,salavrakos_bell_2017}. 
A particularly appealing class of multipartite Bell inequalities is composed of few-body correlators (often, two-body), as these are the interactions that occur in Nature, either in effective Hamiltonians or on those derived from first principles. 
These Hamiltonians, also interpretable as Bell operators, can give an insight on the nonlocality of quantum systems by, for example, looking at their ground state energy~\cite{tura_energy_2017,fadel_bell_2018,wang_two-dimensional_2018,emonts_effects_2024}. 

The intersection of all valid Bell inequalities forms the so-called local (or Bell) polytope~\cite{froissart_constructive_1981}.
Any half-space in the set of correlations can be turned into a Bell inequality by adding a suitable constant. 
There exists an extremal value for such a constant so that the half-space contains the whole Bell polytope.
This is known as the classical bound.
The Bell polytope can be characterized via its vertices or equivalently via its facets.
While the vertices are easy to generate, their number scales exponentially with the number of parties.
The generation of the resulting facets very quickly becomes intractable~\cite{chazelle_optimal_1993,pitowsky_optimal_2001}. 
However, through the use of symmetries or lower-order correlators, expressivity can be traded for computational complexity, thus making the detection of nonlocality in large systems more amenable.

In this work, we focus our efforts on tailoring multipartite Bell inequalities to quantum many-body systems in a one-dimensional geometry. 
Our objective is to characterize the translation-invariant (TI) Bell local polytope, not in the thermodynamic limit~\cite{wang_entanglement_2017}, but with a finite number of $N$ parties.
Specifically, we identify the vertices of the local polytope and, in some cases, we are able to compute its facets.
Finding the classical bound is an instance of a discrete optimization problem.
Many of these problems can be elegantly formulated in the language of tropical algebra~\cite{maclagan_introduction_2021,liu_tropical_2021}. 
Here, we demonstrate that employing a tensor network formalism within the context of tropical algebra provides an effective framework for optimizing Bell inequalities in such a setting. 
Additionally, notions from tropical algebra, like the tropical eigenvector can be utilized to list all the vertices of the corresponding local polytope.
We also establish a connection between the so-called critical graph~\cite{nowak_tropical_2014} and the vertex and facet information of the TI local polytope. 
We show an inclusion relation between the polytope at the thermodynamic limit and the TI local polytopes with finite number of parties $N$.  
Surprisingly, we find that the number of vertices of the TI local polytope is bounded by a finite number independent of $N$.
Moreover, we provide an algorithm to list all vertices that is of complexity $\mathcal{O}(1)$ in $N$, contrary to the naive enumeration, which runs in $\mathcal{O}(\exp N)$. 
All methods proposed in this paper are focused on two-body correlators with nearest-neighbor and next-to-nearest-neighbor interactions but also generalize naturally to a larger interaction range.

The paper is organized as follows: In \Cref{sec:prelim} we give a quick review of Bell nonlocality and we establish the required preliminary concepts in tropical algebra.
For readers familiar with these concepts, this section can be safely skipped.
In \Cref{sec:MathTINNPolytope} we provide the mathematical definition of the translation-invariant nearest-neighbor (TINN) local polytope, interpret its facets using tropical algebra and vertices using graph theory. 
In \Cref{sec:MathTIRPolytope}, we extend the mathematical definition to TI local polytopes with a higher interaction range (TI-$R$) and generalize the connections to tropical algebra and graph theory.
In \Cref{sec:vertices_projected_poly}, we prove that the number of vertices of the TI-$R$ local polytope can be bounded by a constant independent of the system size. 
We also present an algorithm to simultaneously compute its vertices for every finite $N$ and give the result for the TINN case. 
In \Cref{sec:renormTN}, we classify and design TI-$R$ Bell inequalities using the renormalization of tropical tensors.
In \Cref{sec:tropical_eig+eigvec}, we build a connection between the so-called critical graph and faces of the TI-$R$ local polytope.
The code to generate the data and examples used in this paper is available \href{https://gitlab.com/elo_val/ti-bell-inequalities}{online}~\footnotemark[0].

\section{Preliminaries}
\label{sec:prelim}

Throughout the paper, we will use certain abbreviations and specific notation.
The range from $0$ to $r-1$ is denoted $[r] := \{0,\dots,r-1\}$.
For $N$-partite systems, we denote the parties by $A^{(i)}$, where $i\in[N]$ indexes the party and we label them cyclically, i.e.\ $i \mod N$.
For a single party, the index $i$ is omitted.
The bold font, e.g.\ $\boldsymbol{P}$, $\boldsymbol{A}$, $\boldsymbol{s}$, etc.\ refers to multipartite quantities (of $N$ parties in general, but not necessarily).

\subsection{Theory of Bell nonlocality}
\label{subsec:bell_non_locality}

A Bell test is a method to study correlations between subsystems of a larger system \cite{bell_einstein_1964,clauser_proposed_1969}.
Consider a system shared between $N$ parties such that each party performs a measurement on its subsystem.
We assume that each party has the choice between $m$ different measurement bases (also called inputs), and each of them has $d$ outcomes. 
In this work, we will consider the specific case of $d = 2$.
In such case, we say that the scenario of the Bell experiment is $(N,m,2)$ \cite{pironio_lifting_2005, brunner_bell_2014}. 
We denote $\boldsymbol{P}(a_0,\dots,a_{N-1}\vert x_0,\dots,x_{N-1})$ the multipartite probability that each party $i \in [N]$ gets outcome $a_i \in [2]$ given the measurement basis choice $x_i\in[m]$. 

\subsubsection{Local, quantum and non-signaling behaviors}
A behavior is a way to represent the probability distributions describing the Bell experiment. 
It corresponds to a vector in $\mathbb{R}^{(2m)^N}$:
\begin{align}
\label{eq:behavior}
\vec{\boldsymbol{P}} := 
\begin{pmatrix}
    \boldsymbol{P}(0,\dots,0|0,\dots,0) \\
    \boldsymbol{P}(0,\dots,0|0,\dots,1) \\
    \vdots \\
    \boldsymbol{P}(1,\dots,1|m-1,\dots,m-1) \\
\end{pmatrix} \quad.
\end{align}
It encodes the joint probabilities for all possible inputs and outcomes. 

A local deterministic strategy (LDS) is a single-party strategy that can be described by a deterministic map.
There are $2^m$ LDSs and the set is denoted $\LDSset$ \cite{pironio_lifting_2005,brunner_bell_2014}.
For any $s\in \LDSset$, there is a map $f_s:[m]\mapsto [2]$ that associates an outcome to every input. 
The single-party probability distribution associated to the LDS $s$ is given by the Kronecker delta $P_{s}(a\vert x):= \delta_{a,f_s(x)}$.
To a strategy vector $\boldsymbol{s} \in \LDSset^N$, one can associate the local deterministic joint probability: 
$\boldsymbol{P}_{\boldsymbol{s}}(a_0,\dots,a_{N-1} \vert x_0,\dots,x_{N-1}) := \prod_{i=0}^{N-1} P_{s_i}(a_i \vert x_i)$.
In other words, a behavior is local deterministic if it factors as a tensor product $\vec{\boldsymbol{P}}_{\boldsymbol{s}} = \vec{P}^{(0)}_{s_0} \otimes \dots \otimes \vec{P}^{(N-1)}_{s_{N-1}}$, where $\vec{P}^{(i)}_{s_i} \in \mathbb{R}^{2m}$ is the single-party behavior associated to the LDS $s_i$.

According to quantum mechanics, a measurement in the basis $x$ is defined by a positive operator-valued measure (POVM) $M_x$. 
In other words, $M_x$ is a set of Hermitian positive semi-definite operators $\{M_{x,a}\}_{a=0,1}$ satisfying $M_{x,0} + M_{x,1} = \mathbb{I}$.
The single-party probability to measure outcome $a$ given input $x$ is given by the Born rule $P(a|x) = \Tr(M_{x,a}\rho) = \expval*{M_{x,a}}$, where $\rho$ is the quantum state of the system.
The multipartite joint probability is
$\boldsymbol{P}(a_0,\dots,a_{N-1}|x_0,\dots,x_{N-1}) = \expval*{M_{a_0,x_0}\otimes\dots\otimes M_{a_{N-1},x_{N-1}}}$~\cite{born_quantenmechanik_1926}.

A notable property arising in classical and quantum information theory is the non-signaling (NS) property~\cite{popescu_quantum_1994,brunner_bell_2014}. 
Let $\partySubSet := \{i_1,i_2,\dots,i_{|\partySubSet|}\} \subset [N]$ be a subset of parties.
We define the marginal $\boldsymbol{P}^{(\partySubSet)}(a_{i_1},\dots,a_{i_{|\partySubSet|}}|x_0,\dots,x_{N-1})$ by summing over the output $a_i$ of the parties that are not in $\partySubSet$.
A marginal is said to be NS if it is independent of the inputs of the parties that are not in $\partySubSet$. 
In other words, $\boldsymbol{P}^{(\partySubSet)}(a_{i_1},\dots,a_{i_{|\partySubSet|}}|x_0,\dots,x_{N-1}) = \boldsymbol{P}^{(\partySubSet)}(a_{i_1},\dots,a_{i_{|\partySubSet|}}|x_{i_1},\dots,x_{i_{|\partySubSet|}})$.
A distribution is said to be NS if for all subsets $\partySubSet \subset [N]$ the marginal distributions $\boldsymbol{P}^{(\partySubSet)}$ are NS.
For example, distributions associated to LDS or quantum measurements are NS.

The set of all behaviors forms a polytope $\probapoly{(N,m)}$ which has dimension $(2^N-1)m^N$ \cite{pironio_lifting_2005}. 
It is defined by the normalization constraint $\sum_{\{a_i\}_{i=0}^{N-1} = 0,1} \boldsymbol{P}(a_0,\dots,a_{N-1}|x_0,\dots,x_{N-1}) = 1,\forall x_i \in [m], i\in[N]$ and the inequalities $0 \leq \boldsymbol{P} \leq 1$.
The non-signaling constraints together with the normalization constraints define a subspace of dimension $(m+1)^N-1$ in $\mathbb{R}^{(2m)^N}$ \cite{pironio_lifting_2005}. 
The polytope $\nspoly{(N,m)}$ is the intersection of the polytope $\probapoly{(N,m)}$ with the NS subspace. 
The quantum set $\quantumset{(N,m)}$ is the set of behaviors whose components can be written as probabilities obtained from quantum measurements.
Note that the quantum set is convex but not a polytope \cite{pitowsky_range_1986}.
Finally, the local set $\localpoly{(N,m)}$ is a polytope whose vertices are given by the behaviors associated to LDS.
We have the strict inclusions $\localpoly{} \subset \quantumset{} \subset \nspoly{} \subset \probapoly{}$ (we dropped the upper script, for ease of reading) \cite{brunner_bell_2014}.

\subsubsection{LDS in correlator space}

The correlator space is a natural coordinate change of the non-signaling space. 
One-body correlators
are defined as $\expval*{A_x} = P(0|x) - P(1|x)$ and two-body correlators as
$\expval{A_x\cdot B_y}=\sum_{a,b=0}^{1} (-1)^{a+b}{P(a,b|x,y)}$. The generalization to more bodies is straightforward. In total there are $(m+1)^N-1$ correlators, which we will use as coordinates on the non-signaling space. 

For a given $s \in \LDSset$, the correlator takes the value $\expval*{A_{x}}_s = (-1)^{f_s(x)}$, where $f_s$ is the deterministic map introduced in the previous subsection. 
We define the correlator vector as
\begin{align}
\label{eq:singlecorre_without1}
\corrvecNoOne{}_s := 
     \begin{pmatrix}
    \expval{A_{0}}_s\\
    \vdots \\
    \expval{A_{x}}_s \\
    \vdots \\
    \expval{A_{m-1}}_s\\
    \end{pmatrix} 
    \in \mathbb{R}^{m}.
\end{align}
For instance if $m=2$ we have $|\LDSset| =4$. The input-outcome maps are $f_0(x)=0$,  $f_1(x)=x$, $f_2(x)=1-x$ and $f_3(x)=1$. The correlator vectors are
\begin{equation}
\label{eq:s-to-singlecorre}
\begin{split}
& \corrvecNoOne{}_{0} = (1,1)^T, \quad
\corrvecNoOne{}_{1}=(1,-1)^T, \quad  \\
& \corrvecNoOne{}_{2}=(-1,1)^T, \quad
\corrvecNoOne{}_{3}=(-1,-1)^T.
\end{split}
\end{equation}

For LDSs, the bipartite correlators between two parties $A$ and $B$ are given by the products of the correlators~\cite{fine_hidden_1982}:
\begin{align} \label{eq:product_of_correlators}
\expval{A_{x} \cdot B_{y}}_{(s_A,s_B)} = \expval{A_{x}}_{s_A} \cdot
\expval{B_{y}}_{s_B},
\end{align}
where $(s_A,s_B) \in \LDSset^2$, and analogously for more parties.
This means that the $N$-partite behavior associated to an LDS can be written in the correlator space as a tensor product of single-partite correlator vectors:
\begin{align} \label{eq:N-partite_corr_vector} 
\corrvecNoOneNBody{}_{\boldsymbol{s}}
=& \bigotimes_{i=0}^{N-1} 
\begin{pmatrix} 1 \\ 
\corrvecNoOne{}_{s_i}
\end{pmatrix}\in \mathbb{R}^{(m+1)^N},
\end{align}
where $\boldsymbol{s} := (s_0,\dots,s_{N-1}) \in \LDSset^N$ and $s_i$ is the strategy of the party $A^{(i)}$. 
For LDSs, the correlator values of party $A^{(i)}$ are fully specified by the strategy $s_i$, therefore we can drop the label $i$ for the party: $\expval*{A_x^{(i)}}_{s_i} \rightarrow \expval{A_x}_{s_i}$. The entries of $\corrvecNoOneNBody{}_{\boldsymbol{s}}$ are all possible products of correlators between the different parties (including the ``empty product" $1$). 

For instance, in the case $N=m=2$, $\corrvecNoOneNBody{}_{(s,t)}$ can be seen as a matrix
\begin{equation}
    \left[\begin{array}{ c | c c }
    1 & \expval{A_0}_t & \expval{A_1}_t \\
    \hline
    \expval{A_0}_s & \expval{A_0}_s\expval{A_0}_t & \expval{A_0}_s\expval{A_1}_t\\
    \expval{A_1}_s & \expval{A_1}_s\expval{A_0}_t & \expval{A_1}_s\expval{A_1}_t
  \end{array}\right].
\end{equation}
Since the first entry of $\corrvecNoOneNBody{}_{\boldsymbol{s}}$ is always equal to $1$, we will ignore it from now on and view $\corrvecNoOneNBody{}_{\boldsymbol{s}}$ as a vector in $\mathbb{R}^{(m+1)^N-1}$.
Note this is equivalent to the Collins-Gisin notation in correlator space ~\cite{collins_relevant_2004}.

\subsubsection{Local polytope in correlator space}

In the correlator space, the local polytope $\localpoly{(N,m)}$ can be written as the convex hull of LDS correlator vectors ~\cite{fine_hidden_1982}: 
\begin{align}
    \label{eq:vecNmd_without1}
    \localpoly{(N,m)} = \Conv\left\{ \corrvecNoOneNBody{}_{\boldsymbol{s}} : \boldsymbol{s}\in \LDSset^N \right\} \subset \mathbb{R}^{(m+1)^N-1}.
\end{align} 
This polytope $\localpoly{(N,m)}$ is full-dimensional inside this space~\cite{pironio_lifting_2005}. 
Geometrically, the $2^{mN}$ points listed in \Cref{eq:vecNmd_without1} are the \emph{vertices} of the local polytope.

A coefficient vector $\boldsymbol{\alpha} \in \mathbb{R}^{(m+1)^N-1}$ defines the linear function $\boldsymbol{\alpha}\cdot \boldsymbol{q}$, with  $\boldsymbol{q} \in \mathbb{R}^{(m+1)^N-1}$. 
By virtue of convexity, the minimum value on $\localpoly{(N,m)}$ is achieved at a vertex:
\begin{align}\label{eq:classicalBound}
    \beta := \min_{\boldsymbol{q} \in \localpoly{(N,m)}}{\boldsymbol{\alpha} \cdot \boldsymbol{q}}
    = \min_{\boldsymbol{s}\in\LDSset^N}{\boldsymbol{\alpha}\cdot \corrvecNoOneNBody{}_{\boldsymbol{s}}}.
\end{align}
Then every point $\boldsymbol{q} \in \localpoly{(N,m)}$ satisfies the \emph{Bell inequality} 
\begin{equation} \label{eq:BellIneq}
    \boldsymbol{\alpha}\cdot \boldsymbol{q} \geq \beta,
\end{equation}
the number $\beta$ is the so-called \emph{classical bound}.
For convenience, we use the vector notation $(\boldsymbol{\alpha};\beta)$ to encode the inequality.

By the Minkowksi-Weyl theorem~\cite{minkowski_geometrie_2018,farkas_theorie_1902,weyl_elementare_1934,charnes_strong_1958}, there is a finite collection of coefficient vectors $\boldsymbol{\alpha}$ such that the corresponding Bell inequalities define the local polytope. These are called \emph{facet inequalities} or \emph{tight Bell inequalities}. 
For any $\boldsymbol{\alpha}$, the set 
\begin{equation}\label{eq:face}
    \mathcal{H}_{\boldsymbol{\alpha}}:=\{\boldsymbol{q} \in \localpoly{(N,m)} \mid \boldsymbol{\alpha}\cdot \boldsymbol{q} = \beta \},
\end{equation}
where the Bell inequality is achieved is a \emph{face} of the polytope. The facet inequalities are precisely those corresponding to a face of maximal possible dimension (i.e.\ a \emph{facet} of the polytope). 

Computing all tight Bell inequalities is a special case of a general problem in polyhedral geometry: given the vertex description of a polytope, compute its facet description. Algorithms for this problem are available \cite{fukuda_cddlib_2003,fukuda_double_1996}, but they scale exponentially in the number of vertices~\cite{chazelle_optimal_1993}.
In particular, a facet description of the Bell polytope is unknown except for small instances, in particular, small values of $N$ and $m$~\cite{brunner_bell_2014,brunner_partial_2008}.

\subsection{Tropical Algebra}
\label{subsec:TA_intro}

Tropical algebra is a mathematical framework for solving discrete optimization problems and analyzing combinatorial structures.
Let us begin by defining the arithmetic operations of tropical addition $x\oplus y:= \min(x,y)$ and tropical multiplication $x\odot y:= x+y$. 
These operations on real numbers plus infinity give rise to the tropical semiring $(\mathbb{R}_{\infty},\oplus,\odot)$, also referred as the min-sum algebra. 
For a formal exposition, see~\cite{maclagan_introduction_2021,joswig_essentials_2021,akian_max-plus_2006}. Note that the neutral element for tropical addition is $\infty$, and the neutral element for tropical multiplication is $0$.

A \emph{tropical matrix} is a matrix with entries in $\mathbb{R}_{\infty}$; the multiplication of tropical matrices is defined analogously to usual matrix multiplication, with $+$ and $\cdot$ replaced by their tropical analogues:
\begin{align}
\left(F \odot G\right)_{i,j} = \bigoplus_k\left(F_{i,k} \odot G_{k,j}\right) = \min_{k}\left\{F_{i,k} + G_{k,j}\right\}.
\end{align}
In the case $F=G$ is an $n \times n$-matrix, this has an elegant interpretation in terms of graphs: consider the weighted oriented graph with adjacency matrix $F$, i.e.\ if $F_{i,j} \neq \infty$ we draw an arrow from $i$ to $j$ and label it with the weight $F_{i,j}$; if $F_{i,j}=\infty$ we don't draw an edge. 
We will denote this graph by $\Gamma_F$. 
Then the $(i,j)$-th entry of the tropical square $F \odot F$ gives the minimum weight of a path from $(i,j)$ in exactly two steps. 
More generally, the entries of $F^{\odot k}$ give the minimum weight of a $k$-step path from $i$ to $j$.  

\subsubsection{Tropical eigenvalues and eigenvectors} 
\label{subsec:TA_eigva+vec}

Similarly to linear algebra, we can also define the tropical eigenvalue by replacing the traditional multiplication by tropical multiplication. 
If for a tropical $n \times n$ matrix $F$ we can find a vector $v \in \mathbb{R}^n$ and a scalar $\lambda \in \mathbb{R}$ such that
\begin{align}
    \label{eq:eigvector_F}
    F \odot v =\lambda \odot v,
\end{align}
we call $\lambda$ a \emph{tropical eigenvalue} of $F$, and $v$ is called a \emph{tropical eigenvector} of $F$. 
The matrix $F$ is called \emph{irreducible} if the corresponding graph $\Gamma_F$ is strongly connected.
We will always assume this from now on. 
In this case, one can show that $F$ has a unique eigenvalue, which we denote $\lambda(F)$. 
This unique $\lambda$ is equal to the minimum normalized weight (i.e.\ weight divided by the number of edges) of a directed cycle in the graph $\Gamma_F$ and $\lambda$ can be computed in $O(mn)$ time using Karp's Algorithm ~\cite{karp_characterization_1978}, where $m$ is the number of edges in $\Gamma_F$.
The connection between tropical eigenvalue and the optimization of Bell inequalities can be found in \Cref{subsec:Tropical_TINN} and \Cref{subsec:Tropical_TINNN}.

\subsubsection{Kleene plus} 
\label{subsec:kleene_plus}
Let $F'=F-\lambda(F)$ be the normalization of $F$.
This substitution ensures that the eigenvalue is $\lambda(F') = 0$.
All cycles in $F'$ have non-negative weight, and there exists a cycle of weight $0$. 
Let us define the \textit{Kleene plus} of the matrix $F'$ as follows:
\begin{align}
    \label{eq:kleene_plus_F'}
    F'^{+ k}=F' \oplus F'^{\odot 2} \oplus F'^{\odot 3} \oplus \dots \oplus F'^{\odot k}.
\end{align}
The entry $[F'^{+k}]_{i,j}$ gives the minimum weight of all paths of length \emph{at most} $k$ from node $i$ to $j$ in $\Gamma_F$. 
The sequence $\left(F'^{+k}\right)_k$ stabilizes at $k=n$ (recall $n$ is the size of the matrix $F$); we will denote the matrix $F'^{+n}$ simply by $F'^+$. 
Its $(i,j)$-th entry denotes the minimum weight of any path from $i$ to $j$ of any length.
In \Cref{subsec:tropical_power_num}, the Kleene plus is used to classify Bell inequalities.
 
\subsubsection{Computing a tropical eigenvector} 
\label{subsec:pre:trop_eigvec}
One can show that one of the diagonal entries of $F'^+$ is equal to $0$.
For such a diagonal entry $[F'^{+}]_{j,j}=0$, let $v=F'^{+}_{\cdot j}$ be the $j$-th column vector of the matrix $F'^{+}$. 
One can show that $v$ is a tropical eigenvector of $F$ with tropical eigenvalue $\lambda$~\cite{maclagan_introduction_2021}.
Alternatively, as explained in~\cite{hogben_handbook_2006}, computing an eigenvector of $F$ is a single source shortest path problem, which can be solved in $O(nm)$ time, where $m$ is again the number of edges in $\Gamma_F$.
The tropical eigenvector is essential in constructing the critical graph shown in \Cref{fig:example_TINN_min_graph}. 

\subsubsection{Critical graph} 
\label{subsec:pre:crit_graph}

We call a closed path \textit{critical} if its normalized weight is equal to the tropical eigenvalue $\lambda$ of $F$, i.e.\ if it is minimal. 
The \textit{critical graph} of $F$, denoted as $\Gamma_F^{\text{crit}}$, consists of all the nodes and edges that belong to critical closed paths in $\Gamma_F$~\cite{nowak_tropical_2014}. 
Conversely, every closed path in the critical graph will be critical. 
After obtaining the eigenvalue $\lambda$ and one associated eigenvector $v$, the critical graph $\Gamma_F^{\operatorname{crit}}$ can be obtained in $O(m)$ additional time ~\cite{hogben_handbook_2006}. 
We explain now how to do this.

For the row index $k \in [n]$,
 \Cref{eq:eigvector_F} implies that
		\begin{align}
             \label{eq:eigenvector_entry}
			\min_{\ell \in [n]} (F_{k,\ell} +v_\ell)-v_k=\lambda(F).
		\end{align}
We first construct a subgraph $\Gamma^v_F$ by drawing an edge $(k,\ell)$ to every $\ell \in [n]$ for which the minimum in \Cref{eq:eigenvector_entry} is achieved, i.e.\ we draw an edge from $k$ to $\ell$ if $
 F_{k,\ell}+v_{\ell}-v_k=\lambda(F).
 $
Now we can list the strongly connected components of the subgraph $\Gamma_F^v$, which can be found using Kosaraju's algorithm in $O(m)$ time~\cite{hopcroft_data_1983}.  
The union of the strongly connected components forms the critical graph $\Gamma_F^{\operatorname{crit}}$. 
The critical graph will be used to analyze the faces of the projected local polytope, as detailed in \Cref{sec:tropical_eig+eigvec}.

As an example, consider the following $4\times 4$ matrix: 
\begin{align}
\label{eq:maF_TINN}
F=
    \begin{pmatrix}
         2 & 4 &-4 &-2 \\
         0 &2 & -2 & 0\\
        0 & -2 & 2 & 0 \\
        -2 & -4 &  4 & 2
    \end{pmatrix}.
\end{align}
The tropical eigenvalue of $F$ is $\lambda(F)=-2$ and $v=(0,2,2,0)^T$ is a tropical eigenvector. 
The graph $\Gamma_F$ is shown in \Cref{fig:example_TINN_min_graph}.
To construct $\Gamma_F^v$, we connect each pair $(k,\ell)$ for which $(F_{k,\ell} +v_\ell)-v_k=-2$. 
The result is depicted in red in \Cref{fig:example_TINN_min_graph}. 
Since $\Gamma_F^v$ is strongly connected, $\Gamma_F^{\operatorname{crit}}$ and $\Gamma_F^v$ are the same in this example.

\begin{figure}[h!]
    \centering
    \includegraphics[width=0.5\textwidth]{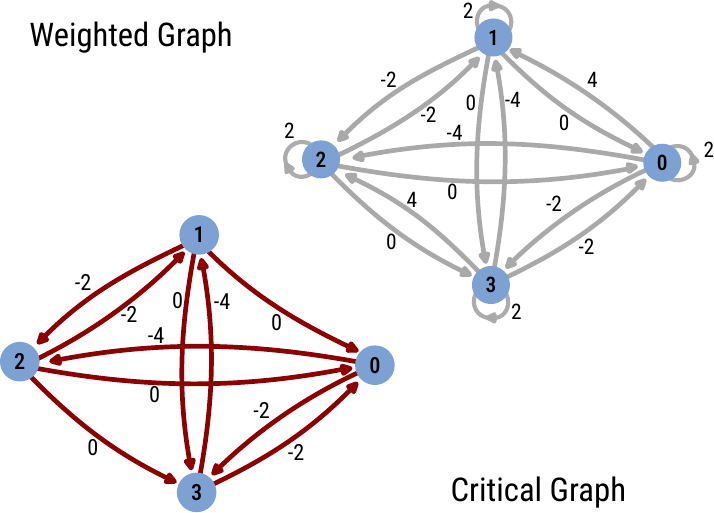}
    \caption{Weighted directed graph $\Gamma_F$ (gray color) and the critical graph $\Gamma_F^{\operatorname{crit}}$ (red color). The critical graph $\Gamma_F^{\operatorname{crit}}$ is strongly connected and forms a subgraph of $\Gamma_F$.
    All the closed paths on $\Gamma_F^{\operatorname{crit}}$ are critical and their normalized weight is $-2$, e.g.\ the closed path $(1,2,3)$. }
    \label{fig:example_TINN_min_graph}
\end{figure}

The critical graph $\Gamma_F^{\operatorname{crit}}$ has $11$ simple cycles: 
\begin{equation} \label{eq:11simpleCycles}
    \begin{split}
       &(0, 3, 1, 2), (0, 3, 1), (0, 3), (0, 2, 3, 1), (0, 2, 3), \\
       &(0, 2, 1, 3), (0, 2, 1), (0, 2), (1, 3), (1, 2, 3), (1, 2).
    \end{split}
\end{equation}
These are precisely the simple cycles of minimal normalized weight in the original graph.

\subsubsection{Sequence of tropical powers and cyclicity}
\label{subsec:pre:trop_power+cyclicity}
In \Cref{subsec:graphTINN,sec:vertices_projected_poly}, we will interpret the classical bound of a translation-invariant Bell inequality as the \emph{tropical trace} of a matrix power. 
The tropical trace of a matrix $F$ is equal to $\operatorname{tropTr}(F):= \bigoplus_i{F_{i,i}}$, the minimal diagonal entry of $F$. Thus, the minimal weight of a cycle of length $N$ in $\Gamma_F$ is given by $\operatorname{tropTr}(F^{\odot N})$. This leads to the formula 
\begin{equation}
    \lambda(F)=\lim_{N \to \infty}{\operatorname{tropTr}(F^{\odot N})/N}
\end{equation}
for the tropical eigenvalue, which is reminiscent of the iterated power algorithm for computing the largest eigenvalue of a matrix (dividing by $N$ is a ``tropical $N$-th root")~\cite{spalding_min-plus_1998,nowak_tropical_2014}. 
In fact, infinitely many of the entries in the sequence $(\operatorname{tropTr}(F^{\odot N})/N)_{N=1,2,\dots}$ are already equal to the limit $\lambda(F)$.
More precisely, the entries $\operatorname{tropTr}(F^{\odot N})/N$ for which $F$ contains a minimal weight cycle of length $N$.
The other entries can be larger, but they will converge to $\lambda(F)$ when $N \to \infty$. 
We will see this pattern again when considering Bell inequalities for TI systems with a finite number of sites.
One can also consider the sequence of matrices $(F^{\odot N})_{N=1,2,\ldots}$.  
According to tropical algebra~\cite{spalding_min-plus_1998, nowak_tropical_2014}, if $F$ is an irreducible matrix, there exist two integers $\sigma$ and $N_0$, such that
\begin{equation} \label{eq:stabilization}
    F^{\odot N+\sigma} = \lambda(F)^{\odot \sigma} \odot F^{\odot N},\quad \forall N \geq N_0.
\end{equation} 
In other words, the sequence $(F'^{\odot k})_{k=1,2,\ldots}$, where $F' = F- \lambda(F)$, eventually becomes periodic with period $\sigma(F)$. 
The number $\sigma=\sigma(F)$ is known as the \emph{cyclicity} of the matrix $F$~\cite{nowak_tropical_2014}. 
It can be computed from the critical graph $\Gamma_F^{\operatorname{crit}}$ as follows: for every strongly connected component of $\Gamma_F^{\operatorname{crit}}$, compute the greatest common divisor of the lengths of all simple cycles. 
Then the cyclicity of $F$ is equal to the least common multiple of these gcd's.
The aforementioned equality 
\begin{equation} \label{eq:tropTrIsLambda}
    \operatorname{tropTr}(F^{\odot N})/N=\lambda(F)
\end{equation}
holds for all $N \geq N_0$ that are divisible by $\sigma$. 

In the above example, the cyclicity $\sigma(F)$ is $\gcd\{2,3,4\}=1$ since the critical graph $\Gamma_F^{\operatorname{crit}}$ is strongly connected and has simple cycles of length $2$, $3$ and $4$. 
Additionally, $N_0$ in \Cref{eq:stabilization} is $2$.

In \Cref{sec:renormTN}, tropical powers and cyclicity will be interpreted as a tropical notion of renormalization. 
These concepts will be used to classify Bell inequalities.

\section{Mathematical definition of the TINN local polytope}
\label{sec:MathTINNPolytope}
In this section, we use the previously defined concepts of Bell nonlocality and tropical algebra to study the so-called TINN local polytope. 
It encodes the translation-invariant Bell inequalities of a 1D system that only contains nearest-neighbor correlators (cf.\ left panel of \Cref{fig:system_sketch}).
In \Cref{sec:MathTIRPolytope}, the case of longer interaction range (right panel in \Cref{fig:system_sketch}) will be considered.

\begin{figure}
    \centering
    \includegraphics[width=\columnwidth]{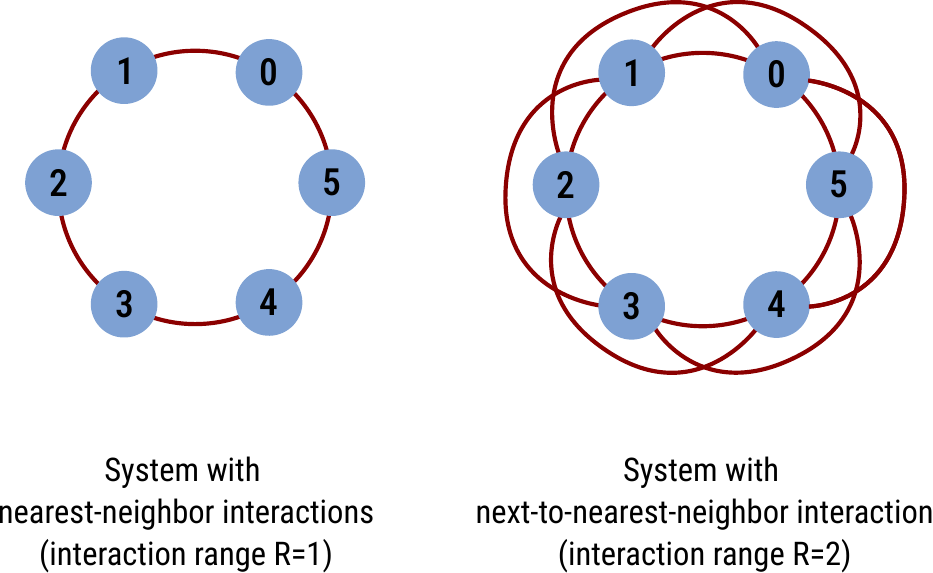}
    \caption{Sketch of the translation-invariant system under consideration.
    Each node represents a party in the Bell scenario and the red lines depict allowed interactions, i.e.\ correlators allowed in the Bell inequality.}
    \label{fig:system_sketch}
\end{figure}

\subsection{Translation-invariant Bell inequalities}
\label{subsec:TIN_Bell_Ineq}

Since the dimension and the number of vertices of the local polytope $\localpoly{(N,m)}$ grows exponentially with the number of parties, the description of the polytope for large $N$ is intractable. 
Therefore we will restrict to Bell inequalties with a specific structure ~\cite{bancal_looking_2010, tura_translationally_2014, tura_detecting_2014, tura_energy_2017, tura_nonlocality_2015}. 
A general Bell inequality [cf.\ \Cref{eq:BellIneq}] can be spelled out as 
\begin{equation} \label{eq:BellIneqInCoords}
    \sum{\alpha^{i_1\dots i_k}_{x_1\dots x_k} \expval*{A^{(i_1)}_{x_1} \cdots A^{(i_k)}_{x_k}}} \geq \beta.
\end{equation}
To start with, we will
study Bell inequalities that
\begin{enumerate}
    \item only involve single-body correlators, and two-body correlators of \emph{nearest neighbors} (NN). That is, the only nonzero coefficients in \Cref{eq:BellIneqInCoords} are $\alpha^{i}_x$ and $\alpha^{i,i+1}_{x,y}$ (where $i+1$ is to be considered modulo $N$). 
    In such a case, we say that the interaction range is $R=1$.
    \item are \emph{translation-invariant} (TI), i.e.\ 
    \begin{align}
        \alpha^{0}_x = \alpha^{1}_x = \dots = \alpha^{N-1}_x, \\
        \alpha^{0,1}_{x,y} = \alpha^{1,2}_{x,y} = \dots = \alpha^{N-1,0}_{x,y}.
    \end{align}
\end{enumerate}

Consider the projection map $\TINNproj$ (the index 1 stands for $R = 1$) that takes as input a correlator vector $\corrvecNoOneNBody{} \in \mathbb{R}^{(m+1)^N-1}$, and outputs the averages of single-body and NN two-body correlators.
Formally, the projection is a vector in $\mathbb{R}^{m+m^2}$ given by
\begin{align} \label{eq:TINN_proj}
\TINNproj\left(\corrvecNoOneNBody{}\right) := 
\frac{1}{N}
\begin{pmatrix}
\sum_{i=0}^{N-1}\corrvecNoOne{^{(i)}} \\
\sum_{i=0}^{N-1}\expval*{\vec{A}^{(i)}\otimes \vec{A}^{(i+1)}} \\
\end{pmatrix},
\end{align}
where $\expval*{\vec{A}^{(i)}\otimes \vec{A}^{(i+1)}}$ denotes the vector of two-body correlators whose entries are given by $\expval*{A^{(i)}_x \cdot A^{(i+1)}_y}$.

For a fixed LDS $\boldsymbol{s} = (s_0,\dots,s_{N-1})$, \Cref{eq:TINN_proj} can be written as a sum over $N$ terms depending only on the strategies of neighboring parties:
\begin{equation}
\TINNproj\left( \corrvecNoOneNBody{} \right) = \frac{1}{N} \sum_{i=0}^{N-1} \corrvecNoOneNBody{_{[R=1]}}_{(s_i,s_{i+1})}
\end{equation}
by defining the correlator vector 
\begin{align}
\label{eq:LDS_TINN_correvec}
\corrvecNoOneNBody{_{[R=1]}}_{(s_i,s_{i+1})}:= 
\begin{pmatrix}
\frac{1}{2}(\corrvecNoOne{}_{s_i}+\corrvecNoOne{}_{s_{i+1}} )\\
\corrvecNoOne{}_{s_i}\otimes\corrvecNoOne{}_{s_{i+1}}
\end{pmatrix}.
\end{align}
As an example, for $m=2$, the TINN correlator vector is the 6 dimensional vector
\begin{equation}\label{eq:TINN_correlator_example}
\corrvecNoOneNBody{_{[R=1]}}_{(s_i,s_{i+1})} =
\begin{pmatrix}
\frac{1}{2} (\expval{A_0}_{s_i} + \expval{A_0}_{s_{i+1}}) \\
\frac{1}{2} (\expval{A_1}_{s_i} + \expval{A_1}_{s_{i+1}}) \\
\expval{A_0}_{s_i} \cdot \expval{A_0}_{s_{i+1}} \\
\expval{A_0}_{s_i} \cdot \expval{A_1}_{s_{i+1}} \\
\expval{A_1}_{s_i} \cdot \expval{A_0}_{s_{i+1}} \\
\expval{A_1}_{s_i} \cdot \expval{A_1}_{s_{i+1}} \\
\end{pmatrix}.
\end{equation}

The image of the local polytope under this projection $\TINNproj$ is also a polytope, we denote it $\localpoly{(N,m)}_{\TINNproj}$. 
It is now easy to verify that the inequalities defining $\localpoly{(N,m)}_{\TINNproj}$ are precisely the Bell inequalities satisfying the TI and NN conditions. 
As shown in \Cref{fig:projection_polytope}, Bell inequalities that are tight in $\mathbb{R}^{m+m^2}$ might not be tight in the original space $\mathbb{R}^{(m+1)^N-1}$, 
since points outside $\localpoly{(N,m)}$ might be projected inside $\localpoly{(N,m)}_{\TINNproj}$. 
\begin{figure}[h!]
    \centering
    \includegraphics[width=\columnwidth]{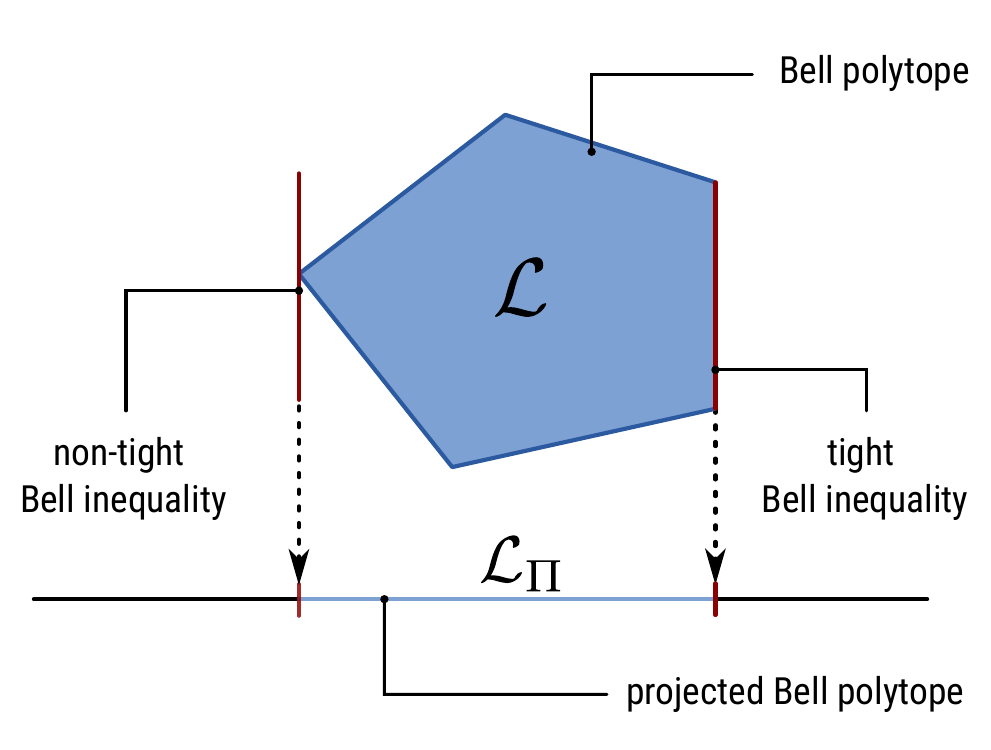}
    \caption{Projection of the local polytope on a subspace. Bell inequalities that are tight in the projected space might not be tight in the original space. Nevertheless, violation of the projected Bell inequalities necessarily implies nonlocality.}
    \label{fig:projection_polytope}
\end{figure}

Importantly, the projected local polytope can be written as the convex hull of projected LDS correlator vectors:
\begin{equation} \label{eq:TINN_polytope}
    \begin{split}
    \localpoly{(N,m)}_{\TINNproj} = \Conv\left\{ \TINNproj\left(\corrvecNoOneNBody{}_{\boldsymbol{s}} \right) : \boldsymbol{s}\in \LDSset^N \right\} \\ 
    = \Conv \left\{ \frac{1}{N} \sum_{i=0}^{N-1}\corrvecNoOneNBody{_{[R=1]}}_{(s_i,s_{i+1})} \mid \boldsymbol{s} \in \LDSset^N \right\},
    \end{split}
\end{equation}
but note that some of the projected correlator vectors might not be vertices of $\localpoly{(N,m)}_{\TINNproj}$. 
As long as $N$ is at least $3$, the polytope $\localpoly{(N,m)}_{\TINNproj}$ is full-dimensional inside $\mathbb{R}^{m^2+m}$ (see \Cref{app:generalTIpolytope} for more details).

\subsection{Optimizing a Bell inequality using tropical algebra} 
\label{subsec:Tropical_TINN}

In this section, we use tropical algebra to find the classical bound of TINN Bell inequalities.
As explained in \Cref{eq:BellIneq}, a Bell inequality can be expressed by a vector of coefficients $\boldsymbol{\alpha}$.
In particular, a Bell inequality in the TINN subspace is defined by the coefficients $\boldsymbol{\alpha} \in \mathbb{R}^{m+m^2}$. The associated linear functional is then given by
\begin{equation}
    I_{\boldsymbol{\alpha}}(\boldsymbol{s}) := \frac{1}{N}\sum_{i=0}^{N-1}{\boldsymbol{\alpha}\cdot\corrvecNoOneNBody{_{[R=1]}}_{(s_i,s_{i+1})}} ,
\end{equation}
hence the linear function can be written as a sum of bipartite terms:
\begin{equation}
    \label{eq:bellfunctionTINN}
    I_{\boldsymbol{\alpha}}(\boldsymbol{s})= \frac{1}{N}\sum_{i=0}^{N-1} I_{\boldsymbol{\alpha}}(s_i, s_{i+1}),
\end{equation}
where $I_{\boldsymbol{\alpha}}(s,t) := \boldsymbol{\alpha}\cdot\corrvecNoOneNBody{_{[R=1]}}_{(s,t)}$ is the expression of the bipartite linear function. 

The classical bound $\beta$ is obtained by optimizing \Cref{eq:bellfunctionTINN} over all LDSs~\cite{tura_energy_2017}.
However, this optimization is in general very hard as it corresponds to classical spin glasses ~\cite{barahona_computational_1982}, since there are exponentially many LDSs to check. 
This optimization task can be restated in terms of contracting an associated tensor network using tropical algebra ~\cite{hu_tropical_2022}, as we now explain. 
The domain of $I_{\boldsymbol{\alpha}}(s,t)$ is the finite discrete set $\LDSset \times \LDSset$, this function can be encoded into the matrix $F(\boldsymbol{\alpha}) \in \mathbb{R}^{2^m \times 2^m}$ by listing all the values of the bipartite linear function: 
\begin{align}
    \label{eq:encodeF_TINN}
    F(\boldsymbol{\alpha}) _{st} = I_{\boldsymbol{\alpha}}(s,t) = \boldsymbol{\alpha}\cdot\corrvecNoOneNBody{_{[R=1]}}_{(s,t)},
\end{align}
where $s,t \in \LDSset$.  
In the rest of the paper, we will often simply write $F$ instead of $F(\boldsymbol{\alpha})$. 
Note that given a matrix $F(\boldsymbol{\alpha})$, it is possible to find back the corresponding TI Bell coefficients $\boldsymbol{\alpha}$ by solving a system of linear equations.
The classical bound per party is finally obtained by the contraction of the tropical tensor network depicted in \Cref{fig:TN_F}: 
\begin{align}
\label{eq:cb_troptrace}
\beta := \min_{\boldsymbol{s}} I_{\boldsymbol{\alpha}}(\boldsymbol{s}) = \frac{1}{N}\operatorname{tropTr}(F^{\odot N}),
\end{align}
where, again, the tropical trace $\operatorname{tropTr}(F^{\odot N})$ is the minimum diagonal entry of $F^{\odot N}$.
\begin{figure}[h!]
    \centering
    \includegraphics[width=\columnwidth]{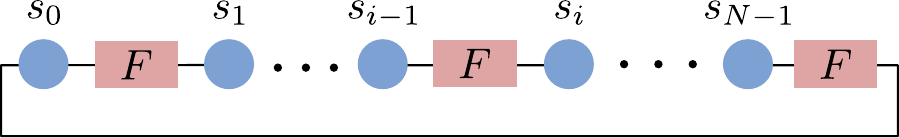}
    \caption{A graphical representation of the tropical tensor network contraction to minimize the linear function $I_{\boldsymbol{\alpha}}(\boldsymbol{s})$ associated to the TINN Bell inequality with coefficients $\boldsymbol{\alpha}$.}
    \label{fig:TN_F}
\end{figure}

As explained in \Cref{subsec:TA_intro}, when the number of parties $N$ grows, the right hand side of \Cref{eq:cb_troptrace} will converge to the eigenvalue of $F$.
Whenever $N$ is divisible by the cyclicity of $F$, we already have an equality $\beta = \lambda(F)$.

The use of tropical algebra reveals a strong connection between Bell inequalities and graph theory. 
One of the main applications of tropical algebra is finding the shortest path in a weighted directed graph. 
Let's consider the graph $\Gamma_F$ with nodes labeled by $s\in \LDSset$ and whose adjacency matrix is ~$F$.
As explained in \Cref{subsec:TA_intro}, the shortest path in $k$ steps from node $s$ to node $t$ in $\Gamma_F$ is given by $[F^{\odot k}]_{st}$.
Therefore the classical bound of the linear function $I_{\boldsymbol{\alpha}}$ can be seen as the shortest normalized closed path of length $N$ in the graph $\Gamma_F$. 

\subsection{The projected local polytope from cycles in the complete graph}
\label{subsec:graphTINN}

In addition to finding the classical bound for a fixed coefficient vector, we can also characterize the entire projected TINN-polytope of local correlations $\localpoly{(N,m)}_{\TINNproj}$ using graph theory. Recall that $F(\boldsymbol{\alpha})$ from \Cref{eq:encodeF_TINN} can be encoded by the graph $\Gamma_F$. This is the complete directed graph with self-loops $K_{2^m}$, where the nodes correspond to $s\in\LDSset$ and the edge $(s,t)$ is labeled by the number $\boldsymbol{\alpha}\cdot\corrvecNoOneNBody{_{[R=1]}}_{(s,t)}$. 
This graph is independent of the system size, and only depends on the number of inputs.  
For our new construction, we consider the same graph $K_{2^m}$, but instead associate the projected correlator vector $\corrvecNoOneNBody{_{[R=1]}}_{(s,t)}$ to the edge $(s,t)$. 
A strategy vector $\boldsymbol{s} \in \LDSset^N$ can be interpreted as a closed path of length $N$ in $K_{2^m}$. 
The point in $\localpoly{(N,m)}_{\TINNproj}$ corresponding to the strategy vector $\boldsymbol{s}$ is exactly the average of the correlator vectors associated to the edges in the closed path $\boldsymbol{s}$. 
Note that a closed path allows node repetitions and if we pass through an edge several times we count it with multiplicity.

As an example, let us consider the scenario $m=2$. 
\Cref{fig:completegraph-example} shows the complete directed graph with four nodes $K_{4}$. 
Each node corresponds to a strategy $s$. 
According to \Cref{eq:s-to-singlecorre}, each strategy is associated to a single body correlator vector $\corrvecNoOne{}_{s}$. 
For each edge $(s,t)$, the NN correlator vector $\corrvecNoOneNBody{_{[R=1]}}_{(s,t)}$ can be computed using \Cref{eq:LDS_TINN_correvec}.
For a system with $N=5$, the strategy vector $\boldsymbol{s} = (0,0,1,3,1)$ is associated to the correlator vector:
\begin{equation} \label{eq:egCorrVec}
\begin{split}
&\TINNproj (\corrvecNoOneNBody{}_{\boldsymbol{s}} ) 
= \frac{1}{5}\bigg( \corrvecNoOneNBody{_{[R=1]}}_{(0,0)} + \corrvecNoOneNBody{_{[R=1]}}_{(0,1)}  \\
& + \corrvecNoOneNBody{_{[R=1]}}_{(1,3)} + \corrvecNoOneNBody{_{[R=1]}}_{(3,1)} + \corrvecNoOneNBody{_{[R=1]}}_{(1,0)} \bigg) \\
& = \frac{1}{5}(3,-1,1,1,1,1)^T,
\end{split}
\end{equation}
where we ordered the entries as in \Cref{eq:TINN_correlator_example}.
\begin{figure}[h!]
    \centering
    \includegraphics[width=\columnwidth]{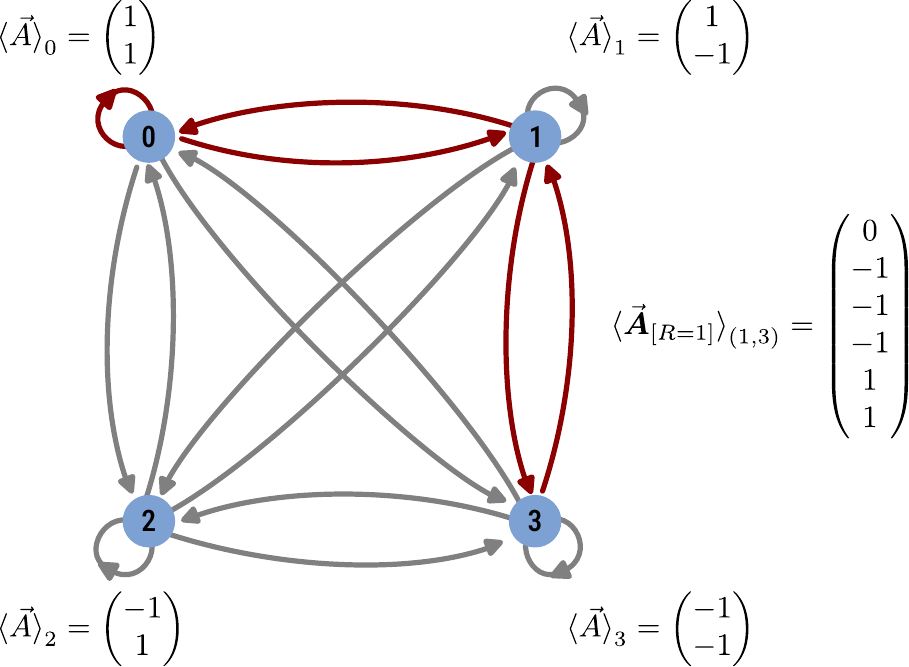}
    \caption{Representation of a $4$-node complete directed graph with self-loops. Each node $s \in \LDSset$ corresponds to an LDS in the $m=2$ scenario. The two-body nearest-neighbor correlator vector $\corrvecNoOneNBody{_{[R=1]}}_{(s,t)}$ can be computed using \Cref{eq:LDS_TINN_correvec}. By averaging $\corrvecNoOneNBody{_{[R=1]}}_{(s,t)}$ along the closed path in red, we get the projection of the vertex of $\localpoly{(5,2)}$ associated to the strategy $\boldsymbol{s} = (0,0,1,3,1)$.}
    \label{fig:completegraph-example}
\end{figure}

The TINN local polytope can be written as a projection of the \emph{normalized closed path polytope} $p_N(K_{2^m})$, which is defined analogously except that to the edges of $K_{2^m}$ we associate linearly independent vectors. Explicitly, for a general graph $\Gamma$, we define $p_N(\Gamma)$ as
\begin{equation} \label{eq:definition_pN}
p_N(\Gamma) := 
    \Conv\left\{ \frac{W(\boldsymbol{s})}{N}  : \boldsymbol{s} ~\text{closed path of } N ~\text{edges on } \Gamma\right\}
\end{equation}
where $W(\boldsymbol{s})$ is the \emph{weight matrix} associated to the closed path $\boldsymbol{s}$ (also known as the strategy vector).
The entry $W_{ij}$ counts the number of times the closed path goes through the edge $(i,j)$.
For example, the closed path $(0,0,1,3,1)$ given in \Cref{fig:completegraph-example} is associated to the $4\times 4$ weight matrix $W_{0,0} = W_{0,1} = W_{1,3} = W_{3,1} = W_{1,0} = 1$ and the other entries are zero.
In what follows, we will denote the
normalized weight matrix by a lower-case letter $w(\boldsymbol{s}):= \frac{1}{\ell(\boldsymbol{s})}W(\boldsymbol{s})$, where $\ell(\boldsymbol{s})$ is the number of edges in the closed path.
We study the polytope $p_N(\Gamma)$ in much more detail in \Cref{app:graph_theory}.

Now, $\localpoly{(N,m)}_{\TINNproj}$ is the image of $p_N(K_{2^m})$ under the linear projection:
\begin{equation} \label{eq:weight matrix_to_corrvec}
     \Phi(w) = \sum_{s,t = 0}^{2^m-1} w_{s,t} \corrvecNoOneNBody{_{[R=1]}}_{(s,t)} ,
\end{equation}
since $\TINNproj(\corrvecNoOne{}_{\boldsymbol{s}}) = \Phi(w(\boldsymbol{s}))$, for all strategy vectors $\boldsymbol{s}$.

According to graph theory~\cite{bondy_graph_2008}, a closed path can be decomposed as a union of simple cycles. 
This leads to writing the weight matrix of a closed path as a convex combination of weight matrices of simple cycles.
This insight can be used to show that for any $N$, the polytope $p_{N}(K_{2^m})$ is included in the \emph{normalized cycle polytope} $p_*(K_{2^m})$, defined (for any graph $\Gamma$) as
\begin{equation} \label{eq:definition_pstar}
p_*(\Gamma) := \operatorname{Conv}\{ w(\boldsymbol{c}) \mid \boldsymbol{c} \text{ simple cycle in } \Gamma\}.
\end{equation}
Note that the polytope $p_*(\Gamma)$ is independent of $N$.
Moreover, if $N$ is divisible by all numbers $1,2,\ldots,2^{m}$, then we have an equality $p_{N}(K_{2^m})=p_{*}(K_{2^m})$.
Both of the above statements are proven in \Cref{thm:periodicityPolytope} in \Cref{app:graph_theory}. 

The image of $p_*(K_{2^m})$ under the projection $\Phi$ gives:
\begin{equation} \label{eq:projectedTINNcv}
\Conv\left\{\frac{1}{\ell(\boldsymbol{c})}\sum_{i=0}^{\ell(\boldsymbol{c})-1} \corrvecNoOneNBody{_{[R=1]}}_{(c_i,c_{i+1})} \mid \boldsymbol{c} \in \mathcal{C}(K_{2^m}) \right\},
\end{equation}
where $\mathcal{C}(K_{2^m})$ is the set of simple cycles in $K_{2^m}$ and $i+1$ is taken modulo $\ell(\boldsymbol{c})$.
We will denote this polytope by $\localpoly{(*,m)}_{\TINNproj}$ and it is independent of $N$.
This implies our first result: 
\begin{theorem}
The TINN local polytope $\localpoly{(N,m)}_{\TINNproj}$ is contained in the polytope $\localpoly{(*,m)}_{\TINNproj}$. If $N$ is divisible by $\lcm{1,\dots,2^m}$, then the inclusion becomes an equality.
\end{theorem}
\begin{proof}
Follows from \Cref{thm:periodicityPolytope} and the fact that $\localpoly{(N,m)}_{\TINNproj}$ and $\localpoly{(*,m)}_{\TINNproj}$ are projections of $p_N(K_{2^m})$ and $p_*(K_{2^m})$.
\end{proof}

The polytope $\localpoly{(*,m)}_{\TINNproj}$ does not only describe the Bell scenario for every $N$ divisible by $\lcm{1,\dots,2^m}$, but also the thermodynamic limit. Indeed, $\localpoly{(*,m)}_{\TINNproj}$ gives a very close approximation of the local polytope $\localpoly{(N,m)}_{\TINNproj}$ for every large $N$ (not necessarily divisible by $\lcm{1,\dots,2^m}$). See also the discussion at the end of \Cref{app:relation_periodicityPolytope}.

For example, consider the scenario $m=2$. 
There are $24$ simple cycles in the graph $K_4$ as shown in \Cref{fig:completegraph-example}: $\mathcal{C} = \{(0),(1),(2),(3),(0,1),(0,2),\dots,(0,3,2,1)\}$. 
Since $\operatorname{lcm}(1,2,3,4)=12$, we have that $\localpoly{(12,2)}_{\TINNproj} = \localpoly{(24,2)}_{\TINNproj} = \ldots = \localpoly{(*,2)}_{\TINNproj}$. 
This polytope has $20$ vertices. 
Each of them can be associated to a simple cycle. Only the simple cycles $(0,1,2,3), (0,2,1,3), (0,3,1,2), (0,3,2,1)$ do not correspond to a vertex; the corresponding vertices of $p_*(K_{4})$ get projected (by $\Phi$) to the interior of the local polytope. 
For $N$ not a multiple of $12$, the polytope might have more vertices; we will analyze this in \Cref{sec:vertices_projected_poly}.

The vertices corresponding to the simple cycles 
$(0)$, $(0,1)$, and $(1,3)$ are given by the respective correlator vectors $v_{(0)}=(1,1,1,1,1,1)^T$, $v_{(0,1)}=(1,0,1,0,0,-1)^T$, and $v_{(1,3)}=(0,-1,-1,0,0,1)^T$. The vector from \Cref{eq:egCorrVec} is a convex combination of these three vertices: $\frac{1}{5}(v_{(0)}+ 2v_{(0,1)} + 2v_{(1,3)})$, compare \Cref{fig:completegraph-example}. 

From the list of vertices of $\localpoly{(*,2)}_{\TINNproj}$, we can compute its list of facets. 
In this case, there are $36$ facets. Using the symmetries described in \Cref{app:symmetry_class}, they are grouped into six classes represented by the following coefficient vectors $(\boldsymbol{\alpha};\beta)$:
\begin{align*}
   (2, \phantom{-}0, \phantom{-}1, \phantom{-}0, \phantom{-}0, \phantom{-}0; -1), \\
   (1, \phantom{-}1, \phantom{-}0, \phantom{-}0, \phantom{-}1, \phantom{-}0; -1), \\
   (2, \phantom{-}0, \phantom{-}1, -1, \phantom{-}1, -1; -2), \\
   (0, \phantom{-}0, \phantom{-}2, -1, \phantom{-}1, \phantom{-}0; -2), \\
   (0, \phantom{-}0, \phantom{-}1, \phantom{-}0, \phantom{-}2, -1; -2), \\
   (0, \phantom{-}0, -2, -1, \phantom{-}1, \phantom{-}0; -2).
\end{align*}
For instance, the fourth coefficient vector corresponds to the tight Bell inequality
\begin{equation}\label{eq:example_face_TINN}
\begin{split}
\frac{1}{N}\sum_{i=0}^{N-1} \big( & 2\expval*{A_0}_{s_i}\expval*{A_0}_{s_{i+1}} - \expval*{A_0}_{s_i}\expval*{A_1}_{s_{i+1}} \\
&+\expval*{A_1}_{s_i}\expval*{A_0}_{s_{i+1}} \big) \geq -2.
\end{split}
\end{equation}

Finally, \Cref{eq:weight matrix_to_corrvec} implies that the number of vertices of $\localpoly{(*,m)}_{\Pi_1}$ is upper bounded by the number of vertices of $p_*(K_{2^m})$,
which is given by the number of simple cycles in $K_{2^m}$:
\begin{align}
\label{eq:logarithmicnumber}
|\mathcal{C}(K_{2^m})| = \sum_{k=1}^{2^m} \frac{2^m!}{k(2^m-k!)} \approx e(2^m-1)! \; .
\end{align}
The number of simple cycles of the complete connected graph form a sequence known as the \emph{logarithmic numbers}, sequence A002104 in OEIS \footnote{J. H. Conway, Logarithmic Numbers, Entry A002104 in The On-Line Encyclopedia of Integer Sequences, https://oeis.org/A002104}.

\section{Longer interaction length}
\label{sec:MathTIRPolytope}

In this section, we consider the TI projection with a longer interaction range $R$ (TI-$R$). 
This means that the projected correlator vector contains only the terms $\expval*{A^{(i)}_x},\expval*{A^{(i)}_xA^{(j)}_y}$, where $j = i+1,\dots,i+R$.
For simplicity, we will review here the case $R=2$ (nearest and next-to-nearest neighbor interactions).
The generalization to any interaction range $R$ is straightforward (see \Cref{app:generalTIpolytope}).

\subsection{TI-2 projection}
The TI-$2$ projection $\TINNNproj$ is similar to the TINN projection in \Cref{eq:TINN_proj}, except that now we consider Bell inequalities for which the only nonzero coefficients are $\alpha_x^{i},\alpha_{x,y}^{i,i+1}$ and $\alpha_{x,y}^{i,i+2}$.
Such projection can be written as
\begin{align} \label{eq:TINNN_proj}
\TINNNproj\left(\corrvecNoOneNBody{}\right) := 
\frac{1}{N}
\begin{pmatrix}
\sum_{i=0}^{N-1}\corrvecNoOne{^{(i)}} \\
\sum_{i=0}^{N-1}\expval*{\vec{A}^{(i)} \otimes \vec{A}^{(i+1)}} \\
\sum_{i=0}^{N-1}\expval*{\vec{A}^{(i)}\otimes \vec{A}^{(i+2)}} \\
\end{pmatrix}.
\end{align}
For a given LDS $\boldsymbol{s} = (s_0,\ldots,s_{N-1})$, \Cref{eq:TINNN_proj} can be written as a sum over $N$ terms 
\begin{equation} \label{eq:TIR_proj_for_LDS}
    \TINNNproj\left(\corrvecNoOneNBody{}\right) = \frac{1}{N}\sum_{i=0}^{N-1} \corrvecNoOneNBody{_{[R=2]}}_{(s_i,s_{i+1},s_{i+2})},
\end{equation}
where $\corrvecNoOneNBody{_{[R=2]}}_{(s_i,s_{i+1},s_{i+2})}$ represents the correlator vector 
\begin{align} \label{eq:LDS_TINNN_correvec}
\begin{pmatrix}
\frac{1}{3}(\corrvecNoOne{}_{s_i}+\corrvecNoOne{}_{s_{i+1}}+\corrvecNoOne{}_{s_{i+2}} )\\
\frac{1}{2}(\corrvecNoOne{}_{s_i} \otimes \corrvecNoOne{}_{s_{i+1}}+\corrvecNoOne{}_{s_{i+1}} \otimes \corrvecNoOne{}_{s_{i+2}}) \\
\corrvecNoOne{}_{s_i} \otimes \corrvecNoOne{}_{s_{i+2}}
\end{pmatrix}.
\end{align}

For example, in the case $m=2$, the strategies $(s_i, s_{i+1}, s_{i+2})=(0,0,3)$ give the correlator vector $\corrvecNoOneNBody{_{[R=2]}}_{(0,0,3)}=(1/3,1/3\,\vert\,0,0,0,0\,\vert\,-1,-1,-1,-1)^T$ using \Cref{eq:s-to-singlecorre,eq:LDS_TINNN_correvec}.
In \Cref{app:generalTIpolytope}, we explain how to to similarly define the correlator vector $\corrvecNoOneNBody{_{[R]}}_{(s_i,\ldots,s_{i+R})}$ for $R=3,4,\ldots$; see in particular  \Cref{eq:LDS_TIR_correvec}.

It is then straightforward to express the $N$-partite TI-$2$ polytope: 
\begin{equation}
    \label{eq:TINNN_polytope}
    \localpoly{(N,m)}_{\TINNNproj} = \Conv\left\{ \TINNNproj \left( \corrvecNoOne{}_{\boldsymbol{s}} \right) \mid \boldsymbol{s}\in\LDSset^N\right\} ,
\end{equation} 
with $\TINNNproj(\corrvecNoOne{}_{\boldsymbol{s})}$ defined as in \Cref{eq:TIR_proj_for_LDS}. 
Note that $\localpoly{(N,m)}_{\TINNNproj} \subset \mathbb{R}^{m+2m^2}$ is full-dimensional when $N>5$ (see \Cref{app:generalTIpolytope} for more details).

\subsection{Optimizing TI-\texorpdfstring{$2$}{2} Bell inequalities using tropical algebra}
\label{subsec:Tropical_TINNN}

Similar to \Cref{subsec:Tropical_TINN}, 
for a given coefficient vector $\boldsymbol{\alpha} \in \mathbb{R}^{m+2m^2}$ for a Bell inequality in the TI-$2$ subspace, the associated linear function is
\begin{equation}
    I_{\boldsymbol{\alpha}}(\boldsymbol{s}) := \frac{1}{N}\sum_{i=0}^{N-1}{\boldsymbol{\alpha}\cdot \corrvecNoOneNBody{_{[R=2]}}_{(s_i,s_{i+1},s_{i+2})}}.
    \label{eq:def_I_TINNN_trop}
\end{equation}
The linear function $I_{\boldsymbol{\alpha}}$ can be written as a sum of tripartite linear functions:
\begin{equation}
    \label{eq:bellfunctionTINNN}
    I_{\boldsymbol{\alpha}}(\boldsymbol{s})= \frac{1}{N}\sum_{i=0}^{N-1}I_{\boldsymbol{\alpha}}(s_i, s_{i+1},s_{i+2}).
\end{equation}
where $I_{\boldsymbol{\alpha}}(s_i, s_{i+1},s_{i+2}) := \boldsymbol{\alpha}\cdot \corrvecNoOneNBody{_{[R=2]}}_{(s_i,s_{i+1},s_{i+2})}$.

Since we have three parties in $I_{\boldsymbol{\alpha}}(s_i, s_{i+1},s_{i+2})$, it is not straightforward to construct a matrix from this tripartite linear function. 
In particular, the matrix has to be square to apply the same methodology. 
However, there is not a unique way to construct the matrix. 
Here, we present three methods giving rise to three different matrices $F,G$ and $H$. 
While the methods are, in principle, equivalent, the different constructions may lead to varying properties. 
In \Cref{subsec:tropical_power_num}, the stabilization properties of these matrices will be connected to the violation of Bell inequalities.

\paragraph{Method 1:} consider a square matrix with more entries, and make sure that only some of them are related to $I_{\boldsymbol{\alpha}}$.
Fill the rest with \enquote{tropical zeros}, i.e.\ $\infty$.
\begin{align}
    \label{eq:tensorF}
    F_{(s,t),(t',u)}=\begin{cases}
    I_{\boldsymbol{\alpha}}(s,t,u) = \boldsymbol{\alpha}\cdot \corrvecNoOneNBody{_{[R=2]}}_{(s,t,u)} & \text{if}\ t=t', \\
    \infty & \text{otherwise,} 
\end{cases}
\end{align}
where $s,t,t',u \in \LDSset$ and $F \in \mathbb{R}^{2^{2m} \times 2^{2m}}$. 
In words, the conditions in \Cref{eq:tensorF} ensure consistency among the LDS.
A party has to commit to a strategy $t$ in each round of the Bell game. 
The conditions ensure that a party chooses the same strategy in pairings with different parties.
As one can see, the matrix $F$ encodes the linear function $I_{\boldsymbol{\alpha}}$, i.e.\ $F_{(s_i,s_{i+1}),(s_{i+1},s_{i+2})} = I_{\boldsymbol{\alpha}}(s_i, s_{i+1},s_{i+2})$.
Note by construction, $F$ is ``tropically sparse". Next, we can compute the classical bound of $I_{\boldsymbol{\alpha}}(\boldsymbol{s})$ using the tropical trace of $F^{\odot N}$ as follows:
\begin{equation}
     \label{eq:TINNN_bound_trF}
    \beta:=\min_{\boldsymbol{s}\in\LDSset^N} I_{\boldsymbol{\alpha}}(\boldsymbol{s})=\frac{1}{N}\operatorname{tropTr}(F^{\odot N}).
\end{equation}
Similar to \Cref{subsec:Tropical_TINN}, $\beta$ converges to the tropical eigenvalue $\lambda(F)$ for $N\to \infty$, and $\beta=\lambda(F)$ whenever $N$ is a multiple of the cyclicity of $F$.
In \Cref{subsec:tropical_power_num} we will perform a detailed analysis of the sequences $(F^{\odot N})_{N=1,2,\ldots}$.

As we will see in the following two methods, it is possible to represent the linear function $I_{\boldsymbol{\alpha}}$ in matrices by grouping two parties together or eliminating every third party.
The encoding of strategies in $F$ is quite sparse, i.e. many matrix elements are $\infty$ to ensure consistency.
Using tropical algebra, we can encode the information more efficiently and possibly change the stabilization characteristics of the matrix.

\paragraph{Method 2:}
we can define a square matrix $G=F^{\odot 2}$, whose entries are
\begin{equation}
   \begin{aligned}
        \label{eq:entry_G}
    G_{(s,t),(u,v)} &=\min_{(t',u')}F_{(s,t),(t',u)}+F_{(t',u),(u',v)}, \\
       &=I_{\boldsymbol{\alpha}}(s,t,u) + I_{\boldsymbol{\alpha}}(t,u,v).
   \end{aligned}
\end{equation}
This corresponds to the physical idea of grouping the two nearest parties into pairs. The strategies of each pair can then be represented as a tuple $\mathbf{s}_i=(s_{2i}, s_{2i+1}),i=0,\dots,\lfloor N/2\rfloor-1$. Then the classical bound of $I_{\boldsymbol{\alpha}}(\boldsymbol{s})$ can be computed by contracting the tensor network via tropical algebra as follows:
\begin{eqnarray}
     \label{eq:TINNN_bound_trG}
    \beta:=\min_{\boldsymbol{s}\in\LDSset^N} I_{\boldsymbol{\alpha}}(\boldsymbol{s})=\frac{1}{N}\operatorname{tropTr}(G^{\odot{\lfloor N/2\rfloor}}\odot T),
\end{eqnarray}
where $T$ is the tail matrix defined as 
\begin{equation}
\label{eq:tail_F2}
T :=  
\begin{cases}
\mathbb{I}_{\mathrm{trop}} &\;,\; \text{if } N = 0 \mod 2 \\
F &\;,\; \text{if } N = 1 \mod 2 
\end{cases}
\end{equation}
with the tropical identity
\begin{align}
   \mathbb{I}_{\mathrm{trop}} := \left(
		\begin{array}{cccc}
			0&\infty&\dots&\infty\\
			\infty&\ddots&\ddots&\vdots\\
			\vdots& \ddots& \ddots& \infty\\
			\infty&\dots& \infty &0
		\end{array}
		\right). \notag
\end{align}
Here, we add the tail matrix $T$ to account for leftover parties when the total number of parties $N$ is odd.

\paragraph{Method 3:} 
similarly to the above Method $2$, we can define $H=F^{\odot 3}$. 
This approach is inspired by the first step of the exponentially faster solution of classical bound presented in~\cite{tura_energy_2017}. 
The entries of $H$ are 
\begin{equation}
    \begin{aligned}
        \label{eq:H_entry}
    &H_{(s,t),(v,w)}=\min_{(u,v')} G_{(s,t),(u,v')}+F_{(u,v'),(v,w)}, \\
    &=\min_{u} I_{\boldsymbol{\alpha}}(s,t,u) + I_{\boldsymbol{\alpha}}(t,u,v)+I_{\boldsymbol{\alpha}}(u,v,w).
    \end{aligned}
\end{equation}

One can verify that the entries of $H$ correspond to the following function by eliminating every third party and grouping every two parties into pairs:
	\begin{align}
		\label{eq:optim-f1}
	\tilde{I}_{\boldsymbol{\alpha}}(\mathbf{s}_j,\mathbf{s}_{j+1}):=\min_{s_{3j+2}}\sum_{i=3j}^{3j+2}I_{\boldsymbol{\alpha}}(s_i, s_{i+1},s_{i+2}),
	\end{align}
 where $\mathbf{s}_j=(s_{3j}, s_{3j+1}),j=0,\dots,q$ and $q:=\lfloor N/3 \rfloor$ since we cannot assume that $N$ is a multiple of $3$. 
 Thus $N=3q+r$ and $r$ can be $0,1,2$. By comparing \Cref{eq:H_entry} with \Cref{eq:optim-f1}, it can be observed that $H=F^{\odot 3}$ precisely represents the procedure of eliminating every third party $3j+2,j\in [q]$ before constructing the matrix $H$ 
 by listing all the values of function $\tilde{I}_{\boldsymbol{\alpha}}$. In details, we can define $H$ using $\tilde{I}_{\boldsymbol{\alpha}}(\mathbf{s}_j,\mathbf{s}_{j+1})$ as follows:
\begin{align}
    \label{eq:tensorH}
    H_{\mathbf{s}_j,\mathbf{s}_{j+1}} =\tilde{I}_{\boldsymbol{\alpha}}(\mathbf{s}_j,\mathbf{s}_{j+1}).
\end{align}

Since the minimization in \Cref{eq:optim-f1} runs on the strategy of the $(3j+2)$-th party, the function $\tilde{I}_{\boldsymbol{\alpha}}(\mathbf{s}_j,\mathbf{s}_{j+1})$ does not depend on $s_{3j+2}$, but only on $\boldsymbol{s}_j = (s_{3j},s_{3j+1})$ and $\boldsymbol{s}_{j+1} = (s_{3j+3},s_{3j+4})$.
The \Cref{fig:eg_IR2} is a visualization of the grouping and eliminating procedure. 
\begin{figure}[h!]
\centering
\includegraphics[width=\columnwidth]{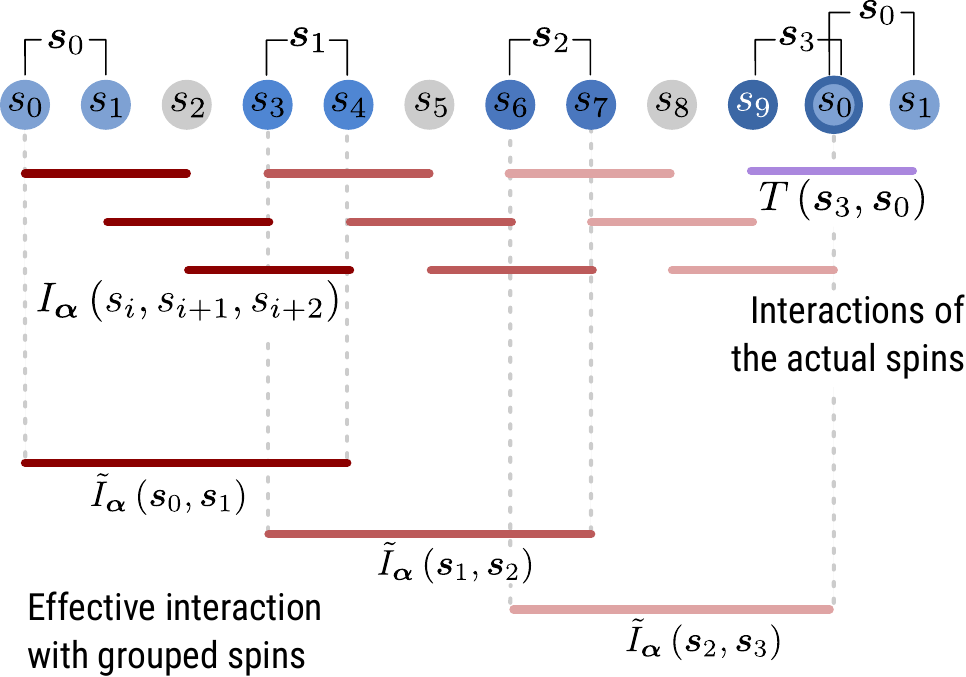}
\caption{A visualization of the grouping and elimination procedure to construct 2-variable functions when $N=10$~\cite{tura_energy_2017}. 
The gray circles correspond to the eliminated parties. 
The parties $3j$ and $3j+1$ with matching colors are paired to form a new super-party $j$. 
Terms of the tripartite function $I_{\boldsymbol{\alpha}}(s_i,s_{i+1},s_{i+2})$ containing the party $3j+2$ (plain lines of the same color) are summed. 
The minimization over $s_{3j+2}$ (gray circle) gives the new 2-variable function $\tilde{I}_{\boldsymbol{\alpha}}(\boldsymbol{s}_j,\boldsymbol{s}_{j+1})$, which can be encoded in a matrix $H$.
It turns out that $H=F^{\odot 3}$, where $F$ is the matrix encoding $I_{\boldsymbol{\alpha}}(s_i,s_{i+1},s_{i+2})$ as shown in \Cref{eq:tensorF}.
The tail function $T(\mathbf{s}_3,\mathbf{s}_0)$ (purple line) can be encoded as $F^{\odot 1}$. 
Following \Cref{eq:troptrace_H}, the classical bound $\beta=\operatorname{tropTr}(H^{\odot 3} \odot F^{\odot 1})$, where we used  $N=3q+r$ with $q=3,r=1$. }
\label{fig:eg_IR2}
\end{figure}

The optimization of $I_{\boldsymbol{\alpha}}(\boldsymbol{s})$ can be written as:
\begin{align}
    \label{eq:eliminateR-th}
    \begin{split}
    & \min_{\boldsymbol{s}\in\LDSset^N} I_{\boldsymbol{\alpha}}(\boldsymbol{s}) \\
    & =\min_{\mathbf{s}_0,\dots,\mathbf{s}_{q}} \frac{1}{N}\left\{\sum_{i=0}^{q-1}\tilde{I}_{\boldsymbol{\alpha}}(\mathbf{s}_i,\mathbf{s}_{i+1})+T(\mathbf{s}_{q},\mathbf{s}_0) \right\},
    \end{split} 
\end{align}
where similarly to Method 2 and~\cite{tura_energy_2017}, we introduce a tail function $T(\mathbf{s}_{q},\mathbf{s}_0)$ to overcome the issue when $N \mod 3 =r\neq 0$. 
The tail function is defined as
\begin{equation}
    \begin{aligned}
            \label{eq:taillocal}
T(\mathbf{s}_{q},\mathbf{s}_0):=&\sum_{i=3q}^{3q+r}I_{\boldsymbol{\alpha}}(s_i,s_{i+1},s_{i+2}),
    \end{aligned}
\end{equation}
and $T(\mathbf{s}_{q},\mathbf{s}_0)$ can be encoded to a square matrix $T \in \mathbb{R}^{2^{2m} \times 2^{2m}}$. This square matrix $T$ is precisely $F^{\odot r}$.
Note that if $r=0$, the tail function is $T=F^{\odot 0}$, which corresponds to the tropical identity matrix $ \mathbb{I}_{\mathrm{trop}}$.

The classical bound $\beta$ of 
$I_{\boldsymbol{\alpha}}$ can be obtained by the contraction of the tropical tensor network as follows:
\begin{align}
    \label{eq:troptrace_H}
    \beta:=\frac{1}{N}\operatorname{tropTr}(H^{\odot q}\odot T).
\end{align}

To summarize, we gave three methods to encode the TI-$2$ linear function $I_{\boldsymbol{\alpha}}(\boldsymbol{s})$ into a matrix $F$, $G=F^{\odot 2}$ or $H=F^{\odot 3}$.
When $N\to \infty$, the classical bound is $\beta=\lambda(F)=\frac{\lambda(G)}{2}=\frac{\lambda(H)}{3}$.

\subsection{The TI-2 local polytope from cycles in the De Bruijn graph}
\label{subsec:DBgraphTINNN}

In \Cref{subsec:graphTINN}, LDSs are represented as nodes on a complete graph $K_{2^m}$, and for each edge $(s_i,s_{i+1})$ we can compute a projected correlator vector $\corrvecNoOneNBody{_{[R=1]}}_{(s_i,s_{i+1})}$. 
In this section, we extend this representation to the TI-$2$ case. 

The projected correlator vector $\corrvecNoOneNBody{_{[R=2]}}_{(s_i,s_{i+1},s_{i+2})}$ depends on three strategies.
Thus, we consider the so-called \emph{De Bruijn graph} $\Gamma_{\db(2^m,2)}$ instead of the complete graph.
The nodes of $\Gamma_{\db(2^m,2)}$ are labeled by a pair of strategies $(s,t)$, and the edges are labeled by triples of strategies, where the edge $(s,t,u)$ links the node $(s,t)$ with the node $(t,u)$. 
In other words, there is an (oriented) edge between two nodes $(s,t)$ and $(t',u)$ if and only if $t=t'$. 
An example of the De Bruijn graph $\Gamma_{\db(2,2)}$ is given in \Cref{fig:ex_DB_graph}. 
It can be thought as a game of dominoes: the node $(0,1)$ can only be followed by a node whose first digit is $1$, i.e.\ either the node $(1,0)$ or $(1,1)$.
Closed paths on a De Bruijn graph can be seen as domino loops, e.g.\ the sequence of nodes $(0,1),(1,0),(0,0)$ forms a closed path on $\Gamma_{\db(2,2)}$.

\begin{figure}[h!]
  \centering
    \includegraphics[width=\columnwidth]{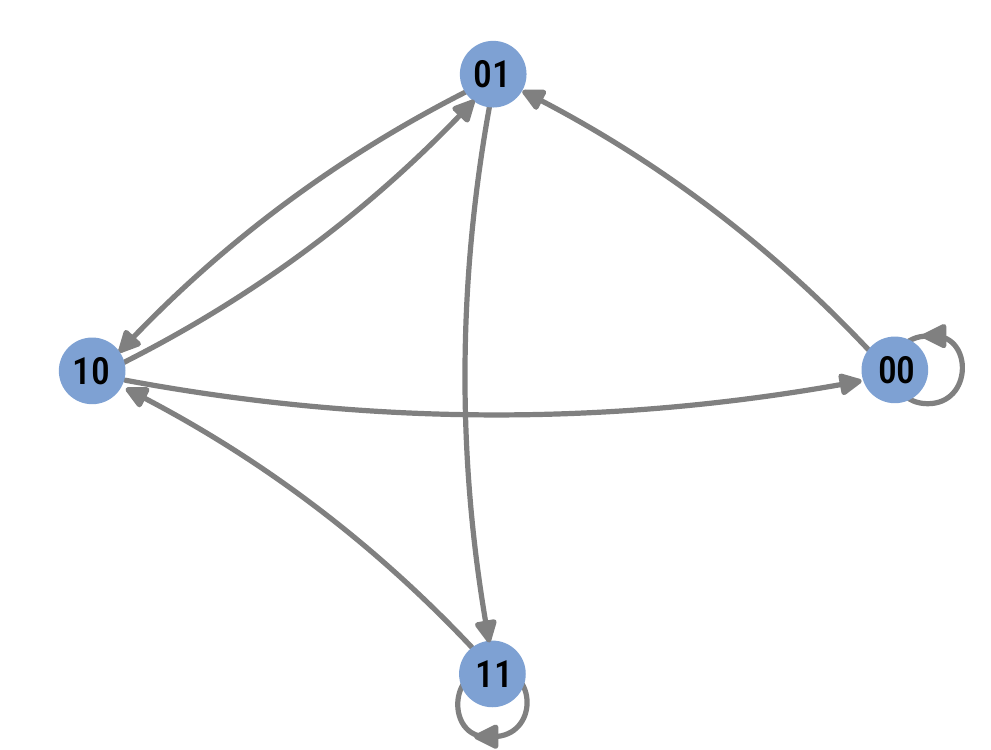}
    \caption{An example of a De Bruijn  graph $\Gamma_{\db(2,2)}$. In a path, the node $(0,1)$ can only be followed by either node $(1,0)$ or node $(1,1)$. The last digit of the node $(0,1)$ must match the first digit of the following node.}   
    \label{fig:ex_DB_graph}
\end{figure}

This same construction can be done for any interaction range $R$: the nodes of the De Bruijn graph $\Gamma_{\db(2^m,R)}$ are labeled by $R$-tuples of strategies and the edges labeled by $R+1$-tuples, where the edge $(s_0,\ldots,s_R)$ links the nodes $\mu=(s_0,\ldots,s_{R-1})$ and $\nu=(s_1,\ldots,s_R)$. We will occasionally denote this edge simply by $(\mu,\nu)$.  
As a specific case, the De Bruijn graph with interaction range $R=1$ is equal to the complete graph, i.e.\  $\Gamma_{\db(2^m,1)} = K_{2^m}$. 

To every edge $(\mu,\nu)=(s_0,\ldots,s_R)$ of the De Bruijn graph, we associate the correlator vector
\begin{equation}
\label{eq:def_corvec_edge}
\corrvecNoOneNBody{_{[R]}}_{(\mu,\nu)} := \corrvecNoOneNBody{_{[R]}}_{(s_0,\ldots,s_{R})},
\end{equation}
as defined in \Cref{eq:LDS_TINNN_correvec} for the case $R=2$ and \Cref{eq:LDS_TIR_correvec} for general $R$.

To connect with the previous subsection: the entries of the matrix $F(\boldsymbol{\alpha})$ in \Cref{eq:tensorF} are nonzero on the positions corresponding to the edges of the De Bruijn graph, and they are obtained by applying the linear map $\boldsymbol{\alpha}$ to \Cref{eq:def_corvec_edge}. This tells us how to generalize Method 1 from \Cref{eq:tensorF} to larger interaction range: if we define a $2^{Rm} \times 2^{Rm}$ matrix $F(\boldsymbol{\alpha})$ with entries 
\begin{equation}
    \label{eq:entry_F_general}
    F(\boldsymbol{\alpha})_{(\mu,\nu)} = \begin{cases}
        \boldsymbol{\alpha} \cdot \corrvecNoOneNBody{_{[R]}}_{(\mu,\nu)} & \text{if } (\mu, \nu) \in \Gamma_{\db(2^{m},R)},\\
        \infty & \text{otherwise,}
    \end{cases}
\end{equation}
then Method 1 generalizes mutatis mutandis.

Similarly to \Cref{subsec:graphTINN}, the TI-$R$ local polytope can be related to the normalized closed path polytope $p_N(\Gamma)$ and normalized cycle polytope $p_*(\Gamma)$ defined in \Cref{eq:definition_pN,eq:definition_pstar}, where $\Gamma$ is taken to be the De Bruijn graph $\Gamma_{\db(2^m,R)}$.
An edge $(\mu,\nu)$ of $\Gamma_{\db(2^m,R)}$ can be associated to $R+1$ strategies $(s_0,\dots,s_R)$.
The TI-$R$ local polytope is related to $p_N(\Gamma_{\db(2^m,R)})$ by the linear projection 
\begin{equation}
\label{eq:map_weight_matrix}
\Phi(w) = \sum_{\mu,\nu = 0}^{n-1} w_{\mu,\nu} \corrvecNoOneNBody{_{[R]}}_{(\mu,\nu)}  ,
\end{equation}
where $n=2^{Rm}$ is the number of nodes in the De Bruijn graph $\Gamma_{\db(2^m,R)}$.
This is the direct generalization of \Cref{eq:weight matrix_to_corrvec}.
Analogously, we define $\localpoly{(*,m)}_{\TIRproj} := \Phi[p_*(\Gamma_{\db(2^m,R)})]$ which is explicitly equal to
\begin{equation}
\label{eq:TINNN_star_convex_hull}
\Conv\left\{\frac{1}{\ell(\boldsymbol{c})}\sum_{i=0}^{\ell(\boldsymbol{c})-1} \corrvecNoOneNBody{_{[R]}}_{(c_i,\dots,c_{i+R})} \mid \boldsymbol{c} \in \mathcal{C} \right\},
\end{equation}
where $\mathcal{C}$ is the set of simple cycles on $\Gamma_{\db(2^m,R)}$ and $\ell(\boldsymbol{c})$ denotes the length of the simple cycle $\boldsymbol{c}$.
These constructions allow us to extend the result from \Cref{subsec:graphTINN} to general $R$.

\begin{theorem}
\label{thm:general_DB_vertices}
The TI-$R$ local polytope $\localpoly{(N,m)}_{\TIRproj}$ is contained in the polytope $\localpoly{(*,m)}_{\TIRproj}$. 
If $N$ is divisible by  $\lcm{1,\dots,n}$, where $n=2^{Rm}$ is the number of nodes in the corresponding de Bruijn graph $\Gamma_{\db(2^m,R)}$, 
then the inclusion becomes an equality.
\end{theorem}
\begin{proof}
Follows from \Cref{thm:periodicityPolytope} and the fact that $\localpoly{(N,m)}_{\TIRproj}$ and $\localpoly{(*,m)}_{\TIRproj}$ are projections of $p_N(\Gamma_{\db(2^m,R)})$ and $p_*(\Gamma_{\db(2^m,R)})$.
\end{proof}

From now on, we say $N$ is \emph{divisible enough} if it is divisible by $\lcm{1,\dots,n}$, where $n$ is the number of nodes in the corresponding De Bruijn graph. 
The formulation in terms of De Bruijn graphs reproduces the result from~\cite{wang_entanglement_2017} where it is shown that TI-$R$ Bell inequalities for $N$ divisible enough can be obtained using irreducible domino loops. 
As an example, the irreducible domino loop $(2,4,3),(4,3,1),(3,1,2),(1,2,4)$ from~\cite[Supplemental Material, Fig. 1]{wang_entanglement_2017} is a simple cycle of length $4$ in the De Bruijn graph $\Gamma_{\db(6,3)}$.  
It can be written in short as $(2,4,3,1)$. 

In the specific case of $m=R=2$, the de Bruijn graph $\Gamma_{\db}(4,2)$ has $120538$ simple cycles, which correspond to the vertices of the normalized cycle polytope $p_*(\Gamma_{\db}(4,2))$.
The local polytope $\localpoly{(*,2)}_{\TINNNproj}$ is the projection of $p_*(\Gamma_{\db}(4,2))$ under the map $\Phi$ from \Cref{eq:map_weight_matrix}; it is the convex hull of $120538$ correlator vectors in $\mathbb{R}^{10}$ as in \Cref{eq:TINNN_star_convex_hull}. It turns out that several simple cycles can give rise to the same correlator vector; the number of distinct vectors is $26213$. Some of these will be vertices of $\localpoly{(*,2)}_{\TINNNproj}$, while others lie on higher-dimensional faces. Using the software PANDA \cite{lorwald_panda_2015}, it is possible to compute the facets and vertices; there are $32372$ facets and $2796$ vertices. Using the symmetries detailed in \Cref{app:symmetry_class}, the facets can be grouped into $2102$ classes~\cite{wang_entanglement_2017} and the vertices into $216$ classes.
For $N$ divisible enough, i.e.\ divisible by $\text{lcm}\{1,\dots,16\}=720720$, $\localpoly{(N,2)}_{\TINNNproj}$ is equal to the polytope $\localpoly{(*,2)}_{\TINNNproj}$ we just described.
Our method also has the advantage of extending the formulation of the TI-$2$ local polytope in terms of the De Bruijn graph to any $N$ (no conditions on the divisibility): then we need to consider closed paths of fixed length $N$ instead of simple cycles.

\section{Vertices for finitely many sites}
\label{sec:vertices_projected_poly}

In this section we analyze the normalized closed path polytopes $p_N(\Gamma)$, and the local polytopes $\localpoly{(N,m)}_{\TIRproj}$, for finite values of $N$. Our first main result is that surprisingly, the number of vertices is bounded by a constant that does not depend on $N$. 

Next, we describe an algorithm that computes the vertices of $p_N(\Gamma)$, and hence of $\localpoly{(N,m)}_{\TIRproj}$, for all values of $N$ at the same time. We deduce from it that $\localpoly{(N,2)}_{\Pi_1}$ has at most $200$ vertices for any $N$.
As a notation convention, we denote $n$ for the number of nodes in the graph $\Gamma$, and $\mathcal{C}$ for the set of simple cycles.

\subsection{Uniform bound for any \texorpdfstring{$N$}{N}}

As a reminder, the polytope $p_N(\Gamma)$ is the convex hull of all the normalized weight matrices of closed paths of length exactly $N$ in $\Gamma$. 
Since a closed path $\boldsymbol{s}$ can be decomposed into simple cycles, its weight matrix can be written as
\begin{equation}\label{eq:simple_cycle_decomposition}
w(\boldsymbol{s}) = \frac{1}{N} \sum_{\boldsymbol{c}' \in \mathcal{C}} a_{\boldsymbol{c}'} W(\boldsymbol{c}'),
\end{equation}
where the $a_{\boldsymbol{c}'}$ are natural numbers and 
$\sum_{\boldsymbol{c}' \in \mathcal{C}}{\ell(\boldsymbol{c}')a_{\boldsymbol{c}'}} = N$, with $\ell(\boldsymbol{c}')$ being the length of the simple cycle. 
Going in the other direction, the right-hand-side of \Cref{eq:simple_cycle_decomposition} 
corresponds to the weight matrix of a closed path, only if the union of the simple cycles in $\{\boldsymbol{c}' \in \mathcal{C} \mid a_{\boldsymbol{c}'} \neq 0\}$ is connected.

\begin{claim} \label{claim:rough_bound}
    If there exist $a_{\boldsymbol{c}'}$ and $a_{\boldsymbol{c}''}$ both larger than the number of nodes $n$ in $\Gamma$, then $w(\boldsymbol{s})$ as defined in \Cref{eq:simple_cycle_decomposition} is not a vertex of $p_N(\Gamma)$.
\end{claim}

\begin{proof}
Suppose that $a_{\boldsymbol{c}'}, a_{\boldsymbol{c}''}>n$. 
We define 
\begin{equation}
a^{(1)}_{\boldsymbol{\gamma}} := 
    \begin{cases}
        a_{\boldsymbol{c}'} + \ell(\boldsymbol{c}'') ,& \text{if } \boldsymbol{\gamma} = \boldsymbol{c}', \\
        a_{\boldsymbol{c}''} - \ell(\boldsymbol{c}') ,& \text{if } \boldsymbol{\gamma} = \boldsymbol{c}'', \\
        a_{\boldsymbol{\gamma}} & \text{otherwise.}
    \end{cases}
\end{equation}
and $a^{(2)}$ analogously with the plus and minus switched. 
It is straightforward to verify that $a^{(i)}_{\boldsymbol{\gamma}}$ are natural numbers for all $\boldsymbol{\gamma} \in \mathcal{C}$ and $\sum_{\boldsymbol{\gamma}} a^{(i)}_{\boldsymbol{\gamma}} \ell(\boldsymbol{\gamma}) = N$.
Moreover, $a_{\boldsymbol{\gamma}}$ and  $a^{(i)}_{\boldsymbol{\gamma}}$ are non-zero for exactly the same $\boldsymbol{\gamma}$, thus the union of the simple cycles $\{\boldsymbol{\gamma} \mid a^{(i)}_{\boldsymbol{\gamma}} \neq 0\}$ is connected, since the union of simple cycles in $\boldsymbol{s}$ is connected.

Therefore, there exist two closed paths $\boldsymbol{s}^{(1)}$ and $\boldsymbol{s}^{(2)}$ of the same length $N$ such that 
\begin{equation} \label{eq:simple_cycle_decomposition_proof}
w(\boldsymbol{s}^{(i)}) = \frac{1}{N} \sum_{\boldsymbol{\gamma} \in \mathcal{C}} a_{\boldsymbol{\gamma}}^{(i)} W(\boldsymbol{\gamma})
\end{equation}
for $i=1,2$.

But now $w(\boldsymbol{s})=\frac{1}{2}\left(w(\boldsymbol{s}^{(1)})+w(\boldsymbol{s}^{(2)})\right)$, which means that $w(\boldsymbol{s})$ is in the convex hull of two other points in $p_N(\Gamma)$, i.e.\ it is not a vertex.
\end{proof}
As a corollary, we obtain the following theorem:
\begin{theorem}
    For any graph $\Gamma$ with $n$ nodes and $|\mathcal{C}|$ simple cycles, the polytope $p_{N}(\Gamma)$ has at most $|\mathcal{C}|\cdot(n+1)^{|\mathcal{C}|-1}$ vertices. 
\end{theorem} 
\begin{proof}
    There are at most $|\mathcal{C}|\cdot(n+1)^{|\mathcal{C}|-1}$ ways to choose the coefficients $a_{\boldsymbol{c}'}$ in \Cref{eq:simple_cycle_decomposition} without two of them being larger than $n$: $|\mathcal{C}|$ choices for which of the coefficients we allow to be large, and $(n+1)^{|\mathcal{C}|-1}$ choices for the values of the other coefficients.
\end{proof}

 Since applying a linear map can only decrease the number of vertices of a polytope, we immediately obtain:
 \begin{cor}
    \label{corr:upper_bound_cycles}
     The projected local polytope $\localpoly{(N,m)}_{\TINNproj}$ has at most $C\cdot(2^m+1)^{C-1}$ vertices, where $C=|\mathcal{C}(K_{2^m})|\approx e(2^m-1)!$ is the number of simple cycles in the complete graph $K_{2^m}$ (crf.\ \Cref{eq:logarithmicnumber}).
 \end{cor}

For instance, for $\localpoly{(N,2)}_{\TINNproj}$ we need to consider the complete graph $K_4$, which has $4$ vertices and $24$ simple cycles. So we can conclude that the number of vertices of $\localpoly{(N,2)}_{\TINNproj}$ is at most $24\cdot5^{23} \sim 3 \cdot 10^{17}$. While this is an astronomically large number, it is a constant independent of $N$. In the next section, we will do a more careful analysis which will lower this bound by several orders of magnitude.

The same result as \Cref{corr:upper_bound_cycles} also holds for larger interaction range $R$. 
The number of vertices of $\localpoly{(N,m)}_{\TIRproj}$ is upper bounded by $C\cdot(2^{Rm}+1)^{C-1}$, where $C$ is the number of simple cycles in the De Bruijn graph $\Gamma_{\db(2^m,R)}$.

\subsection{Computing the vertices of \texorpdfstring{$p_N(\Gamma)$}{P\_N(G)}}

In what follows, we will outline an algorithm to describe the vertices of the polytope $p_N(\Gamma)$. 
Applying it to the complete graph $K_4$ will give the vertices of the TINN local polytope in the case $m=2$ and finite $N$. 

Since the algorithm is rather technical we will only give the main ideas and refer to \Cref{app:graph_theory} for the full description.

In the previous section, we have seen that the vertices of $p_N(\Gamma)$ are given by closed path $\boldsymbol{s}$ of length $N$ [cf. \Cref{eq:definition_pN}] and they are of the form \Cref{eq:simple_cycle_decomposition}, 
where at most one of the coefficients $a_{\boldsymbol{c}'}$ can be larger than $n$. Letting $a_{\boldsymbol{c}}$ be this large coefficient, we can rewrite \Cref{eq:simple_cycle_decomposition} as
\begin{equation} \label{eq:vertices_of_W_N}
   w(\boldsymbol{s}) = w(\boldsymbol{c})+\frac{1}{N}v,
\end{equation}
where $v=\sum_{\boldsymbol{c}' \neq \boldsymbol{c}}a_{\boldsymbol{c}'}\ell(\boldsymbol{c}')(w(\boldsymbol{c}')-w(\boldsymbol{c}))$ is a correction term, whose entries are small. 
Since $w(\boldsymbol{c})$ are the vertices of the normalized cycle polytope $p_*(\Gamma)$, we can interpret this as saying that each vertex of $p_N(\Gamma)$ lies close to a vertex of $p_*(\Gamma)$, with the distance scaling as $\frac{1}{N}$.

We can now fix $\boldsymbol{c}$ and consider for every $N$ the set $\operatorname{err}_{\Gamma}(\boldsymbol{c},N) \subset \mathbb{R}^{n \times n}$ of correction vectors $v$. The crucial fact is that this sequence in $N$ becomes eventually periodic with period $\ell(\boldsymbol{c})$: for $N$ sufficiently large, we have $\operatorname{err}_{\Gamma}(\boldsymbol{c},N) =\operatorname{err}_{\Gamma}(\boldsymbol{c},N+\ell(\boldsymbol{c}))$. This is the content of \Cref{thm:exactVertices} in \Cref{app:graph_theory}. It means that \emph{we can list the vertices of the infinite collection of polytopes $p_N(\Gamma)$ (for $N$ sufficiently large)  using only a finite amount of data:} 
for each simple cycle $\boldsymbol{c}$, and each value $r=0,\ldots,\ell(\boldsymbol{c})-1$, a finite list $\operatorname{err}_{\Gamma}(\boldsymbol{c},r) \subset \mathbb{R}^{n \times n}$. The vertices of $p_N(\Gamma)$ are then of the form \eqref{eq:vertices_of_W_N}, where $v \in \operatorname{err}_{\Gamma}(\boldsymbol{c},  N \operatorname{mod} \ell(\boldsymbol{c}))$. 

In \Cref{app:graph_theory} we prove the aforementioned \Cref{thm:exactVertices}, and moreover give an algorithm to compute the ``error sets" $\operatorname{err}_{\Gamma}(\boldsymbol{c},r)$. The main idea is to view the rescaled polytope $N\cdot p_{N}(\Gamma)$ as the convex hull of a set of lattice points in $N\cdot p_{*}(\Gamma)$. In order to restrict our attention to what happens around a vertex of $N\cdot p_{*}(\Gamma)$, we consider the tangent cone at that vertex. The error set we want to compute is then encoded by the vertices of a certain unbounded polyhedron inside that cone. This polyhedron is given as the convex hull of an explicit but infinite set of points, and only depends on the value of $N$ modulo $\ell(\boldsymbol{c})$. We can reduce to a finite computation using arguments similar to \Cref{claim:rough_bound}, and speed up the computation by triangulating the tangent cone.

In the case of the four nodes complete graph, the number of error points for each simple cycle $\boldsymbol{c}_0$ is summarized in \Cref{tab:errorTerms1}. From this we can compute the number of vertices of $p_N(K_4)$ and $\localpoly{(N,4)}_{\TINNproj}$ for every value $N$; these numbers only depend on the remainder of $N$ modulo $12$, and can be found in \Cref{tab:errorTerms3}. 
In particular, we see that the $p_N(K_4)$ has at most $424$ vertices, and $\localpoly{(N,4)}_{\TINNproj}$ has at most $200$, where both bounds are achieved at for instance $N=19$.

\section{Renormalization of tropical tensors as a method to design Bell inequalities}
\label{sec:renormTN}

In this section, we link our tropical algebra setup to the concept of renormalization~\cite{wilson_renormalization_1974}: the description and study of system changes under rescaling. 
We say the Bell inequality given by a coefficient vector $\boldsymbol{\alpha}$ is invariant under scaling if 
\begin{equation}
    \label{eq:renorm_F_2_def}
    F(\boldsymbol{\alpha})^{\odot 2} = \lambda \odot F(\boldsymbol{\alpha}). 
\end{equation}
Physically, this means that the system behaves in a predictable way if we group the parties pairwise. 
In \Cref{subsec:tropical_renorm}, we present a method to find all Bell inequalities satisfying \Cref{eq:renorm_F_2_def}. 

Note that \Cref{eq:renorm_F_2_def} is the special case of the stabilization equation \Cref{eq:stabilization} where $(\sigma,N_0)=(1,1)$. 
More generally, we can consider the stabilization numbers $(\sigma,N_0)$ associated to the matrix $F(\boldsymbol{\alpha})$ of a given Bell inequality. 
In light of \Cref{eq:tropTrIsLambda}, the classical bound $\beta$ is equal to the eigenvalue $\lambda(F)$ whenever the number of parties is at least $N_0$ and divisible by $\sigma$. 
In \Cref{subsec:tropical_power_num}, we classify the tight TI-2 Bell inequalities by their stabilization numbers. 
Furthermore, we investigate how these numbers are related to quantum violation of the given inequality.

\subsection{Finding stabilized TINN Bell inequalities}
\label{subsec:tropical_renorm}
In this section, we analytically construct TINN Bell inequalities satisfying the stabilization property in \Cref{eq:renorm_F_2_def}, 
where $F(\boldsymbol{\alpha})$ is as defined in \Cref{eq:encodeF_TINN} and $\lambda$ is the eigenvalue of $F$.
Note this is equivalent to \Cref{eq:stabilization} with stabilization numbers $\sigma = N_0 = 1$.

Finding stabilized Bell inequalities means to solve~\Cref{eq:renorm_F_2_def} for $\boldsymbol{\alpha}$ and $\lambda$. 
If we define $F'(\boldsymbol{\alpha},\lambda)$ to be the matrix with entries $F_{i,j}-\lambda$, 
then this equation is equivalent to 
\begin{equation} \label{eq:trop_renorm_f2} 
F'(\boldsymbol{\alpha},\lambda)^{\odot 2} = F'(\boldsymbol{\alpha},\lambda). 
\end{equation}
The condition in \Cref{eq:trop_renorm_f2} can be spelled out as
\begin{align}
    \label{eq:lp_min}
    \min_{k} F'_{i,k}(\boldsymbol{\alpha},\lambda)+F'_{k,j}(\boldsymbol{\alpha},\lambda)=F'_{i,j}(\boldsymbol{\alpha},\lambda), \quad \forall i,j. 
\end{align}
We can first consider a relaxation of this problem: the solutions of \Cref{eq:lp_min} satisfy the linear constraints
\begin{align}\label{eq:linearConstraints}
F_{i,k}(\boldsymbol{\alpha},\lambda)+F_{k,j}(\boldsymbol{\alpha},\lambda) \leq F_{i,j}(\boldsymbol{\alpha},\lambda), \quad \forall i,j,k.
\end{align}
The set of pairs $(\boldsymbol{\alpha},\lambda)$, for which \Cref{eq:linearConstraints} holds, forms a polyhedral cone. 
The solutions to \Cref{eq:trop_renorm_f2} consist of certain faces of this cone; we describe how to compute them in \Cref{Appendix:renormComp}. The computationally most difficult step is computing the rays of the cone from the given facet description. 

For example in the TINN scenario with $m=2$ and no single-body correlators, the matrix $F$ has the form
\begin{align}
	F(\boldsymbol{\alpha})=	\begin{pmatrix}
		\begin{array}{cccc}
			c_0&c_1&-c_1&-c_0\\
			c_2&c_3&-c_3&-c_2\\
			-c_2&-c_3&c_3&c_2\\
			-c_0&-c_1&c_1&c_0
		\end{array} 
		\end{pmatrix},
\end{align} 
where $c_0=\alpha_{00}+\alpha_{01}+\alpha_{10}+\alpha_{11}$, $c_1=\alpha_{00}-\alpha_{01}+\alpha_{10}-\alpha_{11}$, $c_2=\alpha_{00}+\alpha_{01}-\alpha_{10}-\alpha_{11}$ and $c_3=\alpha_{00}-\alpha_{01}-\alpha_{10}+\alpha_{11}$. 
Note that to show $F(\boldsymbol{\alpha})$ in a neat way, we express its entries, as defined in \Cref{eq:encodeF_TINN}, in terms of $c_i$ which are linear combinations of $\boldsymbol{\alpha}$.

The solutions are $F'_1$ with $(c_0,c_1,c_2,c_3,\lambda)=(-1,-1,-1,-1,-1)$ and $F'_2$ with $(c_0,c_1,c_2,c_3,\lambda)=(-1,1,1,-1,-1)$. 
Furthermore, for any real number $\tau$, $\tau F'_1$ and $\tau F'_2$ also belong to the set of solutions $\{F': F'^{\odot 2} = F'\}$. 
The two inequalities corresponding to $F'_1$ and $F'_2$ are 
$\sum_i \expval*{A^{(i)}_0A^{(i+1)}_0} \leq N$ and $\sum_i \expval*{A^{(i)}_1A^{(i+1)}_1} \leq N$. 
If we allow single-body correlators, we find four additional renormalization-invariant Bell inequalities: $-N \leq \sum_i \expval*{A^{(i)}_0} \leq N$ and $-N \leq \sum_i \expval*{A^{(i)}_1} \leq N$. All of these inequalities are rather trivial.

In the more complex TI-$2$ case, the algorithm does not reach completion. Some preliminary computations indicate that applying this algorithm to fractal lattices may generate intriguing Bell inequalities.

\subsection{Bell inequalities classification with tropical power and Kleene plus}
\label{subsec:tropical_power_num}

For TI-$2$ Bell inequalities, we investigate the relationship between their stabilization properties and their quantum violation. 
The classification via stabilization properties gives us a structured way to explore different Bell inequalities.
As explained in \Cref{subsec:Tropical_TINN} and \Cref{subsec:Tropical_TINNN}, a facet of the TI-$R$ local polytope $\localpoly{(N,m)}_{\Pi_{R}}$ is fully described by the coefficient vector $\boldsymbol{\alpha} \in \mathbb{R}^{m+Rm^2}$ and the classical bound $\beta$. 
Equivalently, the facet can be encoded into a matrix $F$ as described in \Cref{eq:tensorF}.

According to \Cref{eq:stabilization,eq:TINNN_bound_trF}, the classical bound scales with the number of parties $N$ as follows
\begin{equation}\label{eq:local_bound_scaling}
    (N+\sigma)\beta_{N+\sigma} = \lambda^{\odot \sigma} + N\beta_N ,\quad \forall N\geq N_0,
\end{equation}
where $\beta_N$ denotes the classical bound for $N$ parties and $\sigma$ is the cyclicity of $F$.
Knowing $\beta_N$ for $N = N_0,\dots, N_0+\sigma-1$, the classical bound can be easily computed for any $N\geq N_0+\sigma$ using \Cref{eq:local_bound_scaling}.
This equation can be interpreted as a renormalization procedure: $\sigma$ parties are grouped together and the rest is (tropically) multiplied by a constant $\lambda^{\odot \sigma}$.

For the rest of this section, we will focus on the polytope $\mathcal{L}_{\TINNNproj}^{(*,2)}$.
For all $32372$ facets, we computed $N_0$ and the cyclicity $\sigma$. 
There are $12$ facets satisfying $F^{\odot 4}=\lambda(F) \odot F^{\odot 3}$ (i.e.\ $\sigma = 1, N_0 =3$), these $12$ Bell inequalities can be classified into two symmetry classes:
\begin{equation}
\begin{split}
& (\boldsymbol{\alpha}_1;\beta_1)=(2, 0, 1, 0, 0, 0, 0, 0, 0, 0; -1), \\
& (\boldsymbol{\alpha}_2;\beta_2)=( 1, 1, 0, 1, 0, 0, 0, 0, 0, 0; -1).
\end{split}
\end{equation}
In these symmetry classes, we consider permutations of the inputs, outputs and parties as equivalent. 
For further details, we refer to \Cref{app:symmetry_class}.
The maximal value of $N_0$ is $26$. There are $8$ facets for which $N_0=26$, they all have $\sigma=1$ and belong to the same symmetry class:
\begin{equation}
    (\boldsymbol{\alpha};\beta)=(4, 0, 2, 0, 0, -4, 4, 4, -4, 1; -9).
\end{equation}

To investigate the properties of this classification, we compute the quantum violation for different pairs of $(\sigma, N_0)$. 
We find the minimal quantum value by optimizing in a seesaw fashion, adapting both the measurements of the inequality and the quantum state.
In a first step, we optimize the state for a given set of measurements with density matrix renormalization group (DMRG)~\cite{white_density_1992,schollwock_density-matrix_2011}. 
Here, the Hamiltonian corresponds to the Bell operator and the ground state is the state with maximal violation given the measurements.
In a second step, the measurements are adapted.
Our parameterization of the measurements is inspired by~\cite{wang_entanglement_2017}, where the authors propose a family of TI measurements.
The optimal parameters for the measurements are found with the BFGS\cite{shanno_conditioning_1970, goldfarb_family_1970, fletcher_new_1970, broyden_convergence_1970} method.
More details can be found in \Cref{app:DMRG}. 
\Cref{fig:Trop_class_F} shows the number of violated inequalities associated to the topical power classification of matrix $F$ for a system with $N=12$. 
We also investigate the number of violated inequalities of the tropical power of matrix $G$ and $H$ constructed using Method 2 and 3 in \Cref{subsec:Tropical_TINNN}, respectively.
Additionally, we also classified the normalized matrices $F'$ according to the stabilization of the Kleene plus $(F'^{+N})_{N=1,2,\ldots}$ (see right panel of \Cref{fig:Trop_class_F}). 
We find that the number of steps before stabilization is always between $4$ and $8$. 

As shown in~\Cref{fig:Trop_class_F}, the stabilization property allows us to discover Bell inequalities with a non-trivial quantum violation in a structured fashion.
However, a large portion of the inequalities do not show a violation.
Since the optimization problem of finding the quantum violation is non-convex, neither of the optimization methods, DMRG or BFGS, give a convergence guarantee.
In order to understand the quantum violation of the given inequalities in more detail, it would be important to obtain a direct relation between the stabilization number and the existence of a quantum violation, ideally without using variational methods.
We leave this investigation for further research.

\begin{figure*}
\centering
\includegraphics[width=0.8\textwidth]{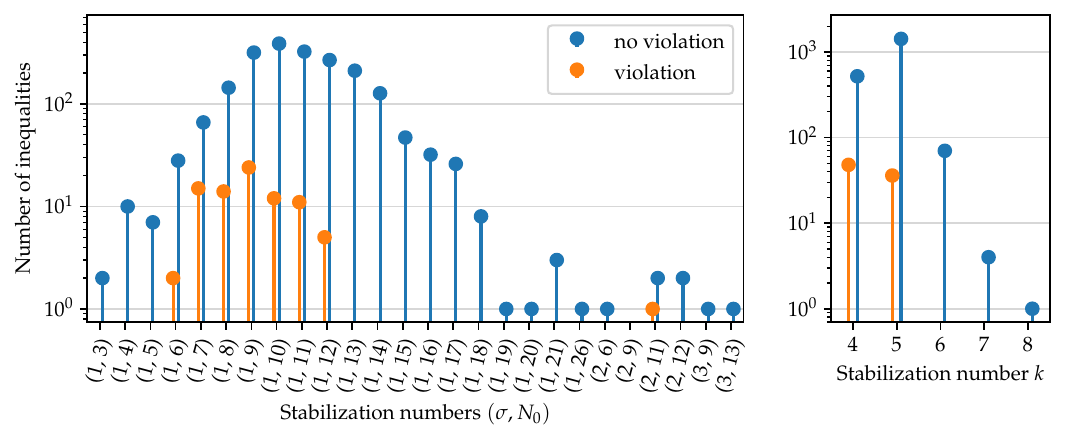}
\caption{Classification of the representative facets of the TI-$2$ local polytope $\localpoly{(*,2,2)}_{\TINNNproj}$ in \Cref{eq:TINNN_star_convex_hull}. On the left: according to the stabilization numbers $(\sigma, N_0)$ of the matrix $F$. On the right: according to the stabilization of the Kleene plus of the matrix $F$.}
\label{fig:Trop_class_F}
\end{figure*}

\section{Tropical eigenvectors and the critical graph}
\label{sec:tropical_eig+eigvec}

In \Cref{subsec:Tropical_TINNN}, we showed that the classical bound of TI-$R$ Bell inequalities can be found using tropical algebra, i.e.\ when the number of parties $N$ is divisible enough, the classical bound is precisely the tropical eigenvalue of the encoded matrix. 
Furthermore, this optimal classical bound can be found by minimizing functions $I_{\boldsymbol{\alpha}}(\boldsymbol{s})$ over all possible strategies of $N$ parties. 
In this section, we will present how to enumerate these optimal strategies by using tropical eigenvectors and the critical graph of the matrices associated to  $I_{\boldsymbol{\alpha}}(\boldsymbol{s})$. From now on, we assume $N$ is divisible enough. 
We also give an example of the TINN case to illustrate this procedure.
Finally in \Cref{subsec:dim_check}, we construct two methods to check the dimension of a given face of the TI-$R$ projected polytope. For concreteness, we explain these methods for 
$R=2$.

\subsection{Optimal strategies encoded in critical graph}
\label{subsec:TA_eigvec_F}

For a given coefficient vector $\boldsymbol{\alpha} \in \mathbb{R}^{m+Rm^2}$ of a Bell inequality in the TI-$R$ subspace, our goal is to list all the possible strategies for $N$ parties such that the classical bound is attained. 
Note that these lists of optimal strategies belong to the face described by this Bell inequality.
We first construct the matrix $F$ corresponding to the TI-$R$ function $I_{\boldsymbol{\alpha}}(\boldsymbol{s})$ using \Cref{eq:entry_F_general}. 
Next, as introduced in \Cref{subsec:TA_intro}, we can construct the critical graph $\Gamma_F^{\operatorname{crit}}$ of the matrix $F$.  
Then finding all the optimal strategy lists of $N$ parties corresponds to listing all the closed paths of length $N$ in the critical graph $\Gamma_F^{\operatorname{crit}}$.
A visual summary of the different expressions of faces of the projected local polytope can be found in \Cref{fig:flowchart}.

\begin{figure}[h!]
    \centering
\includegraphics[width=\columnwidth]{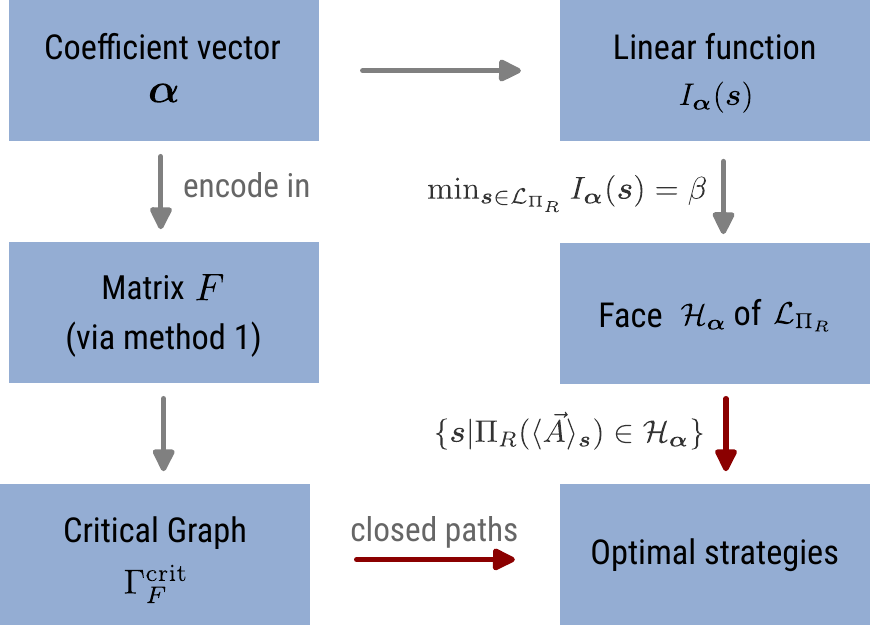}
    \caption{An illustration of the relation between the critical graph and faces of TI-$R$ polytope. The arrows in red color connect the critical graph and the face of the polytope $\mathcal{L}_{\TIRproj}$ together.
    The optimal strategies $\boldsymbol{s}^*$ such that $I_{\boldsymbol{\alpha}}(\boldsymbol{s}^*)=\beta$ are closed paths on $\Gamma_F^{\operatorname{crit}}$.  } 
    \label{fig:flowchart}
\end{figure}

Formally, the connection between optimal strategies and closed paths on the critical graph is given by:

\begin{theorem}
    \label{claim:simcycle_critical}
    The strategy vectors $\boldsymbol{s}^*$ at which the minimum $\min_{\boldsymbol{s}}I_{\boldsymbol{\alpha}} (\boldsymbol{s})=\beta$ is achieved, are precisely the closed paths of the critical graph $\Gamma_F^{\operatorname{crit}}$ of length $N$. 
\end{theorem}
\begin{proof}
Similar to \Cref{subsec:Tropical_TINNN}, minimizing the function $I_{\boldsymbol{\alpha}}$ corresponds to finding the tropical eigenvalue of the matrix $F$ defined in \Cref{eq:entry_F_general} for interaction range $R$. 
The associated graph $\Gamma_F$ is equal to the De Bruijn graph $\Gamma_{\db(2^m,R)}$.
The strategy vectors at which the minimum is achieved correspond to the closed paths of minimal normalized weight in $\Gamma_F$, and as explained in \Cref{subsec:TA_intro}, these are given by the closed paths in the critical graph $\Gamma^{\operatorname{crit}}_F$.
\end{proof}

Next, as an example, we take $m=2$, $R=1$, and encode the coefficient vector $(\boldsymbol{\alpha},\beta)=(0,0,2,-1,1,0;-2)$
from \Cref{eq:example_face_TINN} as
\begin{align}
F(\boldsymbol{\alpha})=
    \begin{pmatrix}
         2 & 4 &-4 &-2 \\
         0 &2 & -2 & 0\\
        0 & -2 & 2 & 0 \\
        -2 & -4 &  4 & 2
    \end{pmatrix},
\end{align}
which is the matrix $F$ in \Cref{eq:maF_TINN} in \Cref{subsec:TA_intro}. 
The critical graph $\Gamma_F^{\operatorname{crit}}$ is shown in \Cref{fig:example_TINN_min_graph}. Recall from \Cref{subsec:graphTINN} that we can associate a correlator vector to every closed path in $K_4$. Now \Cref{claim:simcycle_critical} says that this correlator vector achieves the classical bound if all edges of the path are in the critical graph. For instance for the closed path $\boldsymbol{s}=(0,0,1,3,1)$ the edge $(0,1)$ is not in the critical graph, and indeed the corresponding correlator vector (see \Cref{eq:egCorrVec}) does not achieve the classical bound $\beta=-2$. On the other hand, the closed path $(1,3)$ does lie in the critical graph, and its correlator vector $(0,-1,-1,0,0,1)^T$ does achieve the classical bound: 
\[
(0,0,2,-1,1,0)\cdot (0,-1,-1,0,0,1)^T = -2.
\] 

Taking the thermodynamic limit, we have seen that the polytope $\localpoly{(*,2)}_{\TINNproj}$ has $20$ vertices, corresponding to $20$ of the $24$ simple cycles in $K_4$. Such a vertex lies on the facet given by $\boldsymbol{\alpha}$ if and only if the corresponding simple cycle lies on $\Gamma_F^{\operatorname{crit}}$. 
We have seen that $\Gamma_F^{\operatorname{crit}}$ has $11$ simple cycles, listed in \Cref{eq:11simpleCycles}. However, two of them ($(0,3,1,2)$ and $(0,2,1,2)$) do not correspond to vertices of $\localpoly{(*,2)}_{\TINNproj}$. 
So we conclude that the facet under consideration has $9$ vertices.  

\subsection{Dimension analysis of faces with critical graph}
\label{subsec:dim_check}

Tight Bell inequalities give the unique minimal representation of the local polytope. 
We can attack the problem of finding tight inequalities from two angles: either finding all tight Bell inequalities directly or checking whether a given Bell inequality is tight.
In \Cref{subsec:DBgraphTINNN}, we already took the direct approach.
Here, we take the latter route which is computationally less intensive and check that the face corresponding to a TI-$R$ Bell inequality has dimension $\dim(\localpoly{(*,m)}_{\TIRproj}) - 1 = m+Rm^2-1$.
For sake of concreteness, we consider the TI-$2$ case ($R=2$).   
For example, as explained in \Cref{sec:MathTIRPolytope}, when $N$ is divisible enough, the dimension of TI-$2$ polytope with $m=2$ is $10$.
Thus a TI-$2$ Bell inequality is tight if its associated face has dimension $9$.

\Cref{claim:simcycle_critical} implies that the face defined by $F(\boldsymbol{\alpha})$ is spanned by the vectors $\frac{1}{\ell(\boldsymbol{c})}\sum_{i=0}^{\ell(\boldsymbol{c})-1} \corrvecNoOneNBody{_{[R=2]}}_{(c_i,c_{i+1})}$, where $\boldsymbol{c} = (c_0, c_1 ,\dots)$ runs over the simple cycles the critical graph $\Gamma_{F(\boldsymbol{\alpha})}^{\operatorname{crit}}$.
However, if we only want to compute the dimension of the face, we can do this without having to list all the cycles.
For any face $\mathcal{H}_{\boldsymbol{\alpha}}$, we can consider the vector space $\operatorname{span}(\mathcal{H}_{\boldsymbol{\alpha}} )$ spanned by the face. 
Its dimension is given by  
$\dim(\operatorname{span}( \mathcal{H}_{\boldsymbol{\alpha}} )) = \dim(\mathcal{H}_{\boldsymbol{\alpha}}) + 1$.
This is true as long as all zeroes vector do not lie on the affine span of the face (otherwise $\dim(\mathcal{H}_{\boldsymbol{\alpha}}) = \dim(\operatorname{span}(\mathcal{H}_{\boldsymbol{\alpha}} ))$.
For the polytope $\localpoly{(*,m)}_{\TINNNproj}$, symmetries (see \Cref{app:symmetry_class}) can be used to argue that the all zeroes vector lies in the interior, and therefore not on the affine span of any face.
In other words the Bell inequality given by $\boldsymbol{\alpha}$ is tight if and only if $\dim(\operatorname{span}( \mathcal{H}_{\boldsymbol{\alpha}} )) = m+Rm^2$.

As before, we write $n=2^{2m}$ for the number of nodes in the De Bruijn graph $\Gamma_{\db(2^m,2)}$. Let us write $\mathcal{W}_{\Gamma_{F(\boldsymbol{\alpha})}^{\operatorname{crit}}} \subset \mathbb{R}^{n \times n}$ for the linear space spanned by all normalized weight matrices of closed paths in $\Gamma_{F(\boldsymbol{\alpha})}^{\operatorname{crit}}$ (normalized weight matrices are defined in \Cref{subsec:graphTINN}). 
Then $\operatorname{span}(\mathcal{H}_{\boldsymbol{\alpha}} )$ is the image of $\mathcal{W}_{\Gamma_{F(\boldsymbol{\alpha})}^{\operatorname{crit}}}$
 under the map $\Phi$ defined in \Cref{eq:map_weight_matrix}.
This corresponds to
\begin{equation}
    \label{eq:dim_face_weight_matrix}
    \begin{split}
        \mathcal{H}_{\boldsymbol{\alpha}}=\{\sum_{\mu,\nu} M_{\mu \nu}\corrvecNoOneNBody{_{[R=2]}}_{(\mu,\nu)}\vert M=[M_{\mu \nu}] \in \mathcal{W}_{\Gamma_{F(\boldsymbol{\alpha})}^{\operatorname{crit}}} \},
    \end{split}
\end{equation}
where $\corrvecNoOneNBody{_{[R=2]}}_{(\mu,\nu)}$, as defined in \Cref{eq:def_corvec_edge}, represents the correlator vector assigned on edge $(\mu,\nu)$ of the critical graph $\Gamma_{F(\boldsymbol{\alpha})}^{\operatorname{crit}}$. 
The space $\mathcal{W}_{\Gamma_{F(\boldsymbol{\alpha})}^{\text{crit}}}$ consists of all matrices for which the only nonzero entries are in positions corresponding to the edges of $\Gamma_{F(\boldsymbol{\alpha})}^{\text{crit}}$, and moreover for each $i$ the sum of the entries in row $i$ equals the sum of the entries in column $i$.
These linear equations are formally given in \Cref{eq:def_W_gamma}.
By solving this system of equations, we obtain a matrix $\Lambda$ whose column-span is $\mathcal{W}_{\Gamma_{F(\boldsymbol{\alpha})}^{\operatorname{crit}}}$.

To compute the dimension of $\operatorname{span}(\mathcal{H}_{\boldsymbol{\alpha}} )$, we need to multiply this with the matrix $A$ representing the linear map $\Phi$ (the $n^2$ column vectors of $A$ are the correlator vectors $\corrvecNoOneNBody{_{[R=2]}}_{(\mu,\nu)}$).
Putting everything together, we find that 
\begin{equation}
    \label{eq:dim_face_rank}
    \dim(\mathcal{H}_{\boldsymbol{\alpha}})=\rank(A \cdot \Lambda)-1.
\end{equation}

In the following, we give an example to explain the above method to check the dimension of a face of $\localpoly{(*,2,2)}_{\TINNNproj}$. Let us consider the following specific function $I_{\boldsymbol{\alpha}}(\boldsymbol{s})= \frac{1}{N}\sum_{i=0}^{N-1}I_{\boldsymbol{\alpha}}(s_i, s_{i+1},s_{i+2})$, where
\begin{align}
    \boldsymbol{\alpha}=(-2,-4,-2,2,2,2,1,0,0,1).
\end{align}
The tropical eigenvalue of its corresponding matrix $F$ is $\lambda(F) = -4$, which is exactly the classical bound $\beta=\min_{\boldsymbol{s}}I_{\boldsymbol{\alpha}}(\boldsymbol{s}) = -4$~\cite{wang_entanglement_2017}. 
Moreover, one eigenvector of $F$ is
\begin{align}
     & \frac{1}{3}(0 , -2 , 2 , 0 , 2 , 0 , 4 ,14 , -2 , 8, 0 , -2 , 0 , 10 , 2, 12)^T.\notag 
\end{align}
The critical graph $\Gamma_F^{\operatorname{crit}}$ is as shown in 
\Cref{fig:Nav_IG_crit_graph}. 
There are $43$ simple cycles in \Cref{fig:Nav_IG_crit_graph}, and the cyclicity of $F$ is $\sigma=1$ and $N_0=6$. 
When $N\geq N_0=6$, the optimal strategies are the closed paths on $\Gamma_F^{\operatorname{crit}}$.

Next, we compute the matrix $\Lambda$ by solving the system in \Cref{eq:def_W_gamma}. 
Then for each edge $(\mu,\nu)$, we compute $\corrvecNoOneNBody{_{[R=2]}}_{(\mu,\nu)}$. 
For example, for the edge $(0,3,0)$, the correlator vector $\corrvecNoOneNBody{_{[R=2]}}_{(0,3,0)}$ is $(\frac{1}{3},\frac{1}{3},-1,-1,-1,-1,1,1,1, 1)^T$. 
We construct the matrix $A$ with column vectors $\corrvecNoOneNBody{_{[R=2]}}_{(\mu,\nu)}$.
By using \Cref{eq:dim_face_rank}, we know the dimension is $9$. 
So it is a facet of the TI-$2$ polytope. 

In the present case $R=m=2$, a complete list of tight Bell inequalities is available (\cite{wang_entanglement_2017}, see also \Cref{subsec:DBgraphTINNN}), so we have only confirmed what we already knew. 
If we increase the interaction range or number of inputs, getting a complete facet description quickly becomes intractable. 
However, we can still use our method to compute the dimension of the face corresponding to a given Bell inequality, for instance, one motivated by a renormalization procedure. 
We leave this for future work. 

\begin{figure}[h!]
  \centering
  \includegraphics[width=\columnwidth]{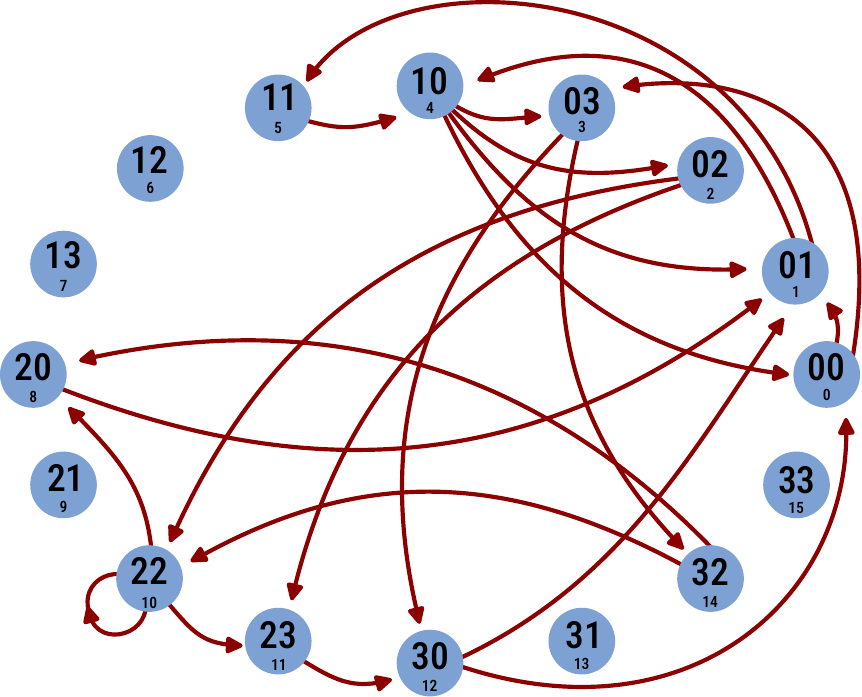}
    \caption{Critical graph $\Gamma_F^{\operatorname{crit}}$ of $F$, which is a subgraph of the De Bruijn graph $\Gamma_{\db(4,2)}$.
    The major numbers in the blue nodes are written in base four to make the identification with the corresponding LDS more natural. 
    The two digits in each node correspond to the two strategies that the two parties chose.
    Their corresponding base $10$ equivalent are written in smaller font below.  
    The only strongly connected component is $\{00,01,02,03,10,11,20,22,23,30,32\}$. 
    }   
    \label{fig:Nav_IG_crit_graph}
\end{figure}

\section{Conclusions and Outlook}
\label{sec:con}
This work establishes a deep connection between the study of nonlocality in many-body systems and combinatorial optimization tasks formulated via tropical algebra.
Even a seemingly simple task such as finding the classical bound for a multipartite Bell inequality composed of two-body correlators is NP-hard because it can be reinterpreted as finding the ground state of a classical spin system~\cite{liu_tropical_2021}.
The tropical algebra formalism combined with the machinery of tensor networks singles out cases in which this problem can be solved efficiently.
Whenever the tensor network contraction can be executed efficiently, we can compute the classical bound in practice.
Here, we have focused on many-body systems in one spatial dimension, which fulfill this property.

Combining tropical algebra with tensor networks is not only useful in the context of nonlocality~\cite{hu_tropical_2022,emonts_effects_2024}.
Also in the realm of many-body systems and the benchmarking of quantum devices, the combination of these tools has proven rather powerful~\cite{liu_computing_2023}.

Our work highlights interesting connections between one-dimensional systems and graph theory, allowing for the characterization of the TI-$R$ local polytope, with constant $R$.
Surprisingly, the number of vertices of this polytope is uniformly bounded. 
This contrasts with the case for a much larger symmetry group, namely permutational invariance, where the number of vertices grows with $N$~\cite{tura_nonlocality_2015}.

Our tensor network formulation motivated us to investigate how the renormalization procedure of tropical tensor networks translates to multipartite Bell nonlocality.
In the TI case, this has a direct interpretation, as the contribution per party to the classical bound.
To make a direct connection between the design of useful inequalities as fixed points of this renormalization procedure, further research is needed.

From a more practical perspective, this formalism does not only allow us to find the classical bound but also to efficiently compute and encode the optimal local deterministic strategies, thereby providing a tool to assess the tightness of a given Bell inequality.
This complements the results of previous investigations~\cite{tura_energy_2017}.

This work suggests several directions for further research.
While it focuses on one-dimensional geometries, the extension to two dimensions is far from trivial.
The main obstacle is the inefficiency of the contraction steps in the general case.
The methodology presented can be pushed to graphs of constant tree-width, e.g.\ hyperbolic or fractal geometries, which can arise in the context of symmetry-protected topological order~\cite{stephen_topological_2021,stephen_subsystem_2019}.

When extending the approach beyond fractal lattices, a stronger focus on numerical methods is required.
In higher dimensions, methods like tensor renormalization group are extremely successful to contract large systems of tensors.
However, these methods rely on truncation (via singular value decomposition) to make the algorithm efficient.
To the best of our knowledge, there is no comparable truncation method for tropical algebra.
Without access to efficient truncations, methods like index slicing or heuristic tensor network contractions could improve overall efficiency~\cite{huang_efficient_2021} and open access to quasi-one dimensional and small two-dimensional systems.

Finally, tropical tensor networks also offer means towards lower-bounding global minima, for instance via the use of epsilon nets~\cite{schuch_matrix_2010,aharonov_efficient_2010}. 
This constitutes not only an interesting alternative to variational methods when computing ground state energies, but also a complementary toolset to convex relaxation techniques which are typically employed~\cite{navascues_bounding_2007, requena_certificates_2023, kull_lower_2024}. 

\begin{acknowledgments}
We thank Albert Aloy and Zizhu Wang for sharing their insights into the Bell polytope, and Jinfu Chen for feedback on an early version of the manuscript. 
T.S. would like to thank Stijn Cambie for pointing out relevant references in the graph theory literature. 
F.M. would like to thank Marianne Johnson and Henry Wynn for helpful discussions on tropical matrix algebra. 
F.M. was partially supported by the Research Foundation – Flanders (FWO), Grant Numbers G0F5921N (Odysseus, Type I) and G023721N, as well as the grant iBOF/23/064 from KU Leuven. 
These grants also supported the research visits of M.H. from Leiden to Leuven, which contributed to this work.
P.E. and J.T. acknowledge the support received by the Dutch National Growth Fund
(NGF), as part of the Quantum Delta NL programme. 
P.E. additionally acknowledges the support received through the NWO-Quantum Technology programme (Grant No. NGF.1623.23.006).
J.T. and E.V. acknowledge the support received from the European Union’s Horizon Europe research and innovation programme through the ERC StG FINE-TEA-SQUAD (Grant No. 101040729). 
T.S. is supported by Research foundation – Flanders
(FWO) - Grant Number 1219723N. 
This publication is part of the 'Quantum Inspire - the Dutch Quantum Computer in the Cloud' project (with project number [NWA.1292.19.194]) of the NWA research program 'Research on Routes by Consortia (ORC)', which is funded by the Netherlands Organization for Scientific Research (NWO).

The views and opinions expressed here are solely
those of the authors and do not necessarily reflect those of the funding institutions. Neither
of the funding institutions can be held responsible for them.
\end{acknowledgments}

\section*{Authors contribution}
\textbf{Conceptualization:} T.S., F.M., J.T. 
\textbf{Methodology:} M.H., E.V., T.S., P.E., F.M., J.T.
\textbf{Software:} M.H., E.V., T.S.
\textbf{Validation:} T.S., P.E., J.T.
\textbf{Formal analysis:} M.H., E.V., T.S., F.M.
\textbf{Data curation:} M.H., E.V., P.E.
\textbf{Writing - Original Draft:} M.H., E.V., T.S., F.M.
\textbf{Writing - Review \& Editing:} P.E., J.T.
\textbf{Visualization:} M.H., E.V., P.E.
\textbf{Supervision:} T.S., P.E., J.T.
\textbf{Project administration:} J.T.
\textbf{Funding acquisition:} T.S., P.E., F.M., J.T.

\appendix

\section{Definition of TI-\texorpdfstring{$R$}{R} polytope for any \texorpdfstring{$R$}{R}} \label{app:generalTIpolytope}

In the main text, the definition of the projected polytope was given for interaction range $R = 1$ and $R=2$.
In this section, we will define the polytope $\localpoly{(N,m)}_{\TIRproj}$ for general interaction range $R\geq1$.

The $N$-partite TI-$R$ polytope $\localpoly{(N,m)}_{\TIRproj}$ is given by
\begin{equation}
    \label{eq:TIR_polytope}
    \Conv\left\{ \TIRproj \left( \corrvecNoOneNBody{}_{\boldsymbol{s}} \right) \mid \boldsymbol{s}\in\LDSset^N\right\},
\end{equation} 
where the TI-$R$ projection reads
\begin{align}
\label{eq:TIR_proj}
\TIRproj\left(\corrvecNoOneNBody{}\right) = 
\frac{1}{N}
\begin{pmatrix}
\sum_{i=0}^{N-1}\corrvecNoOne{^{(i)}} \\
\sum_{i=0}^{N-1}\expval*{\vec{A}^{(i)}\otimes \vec{A}^{(i+1)}} \\
\sum_{i=0}^{N-1}\expval*{\vec{A}^{(i)}\otimes \vec{A}^{(i+2)}} \\
\vdots \\
\sum_{i=0}^{N-1}\expval*{\vec{A}^{(i)}\otimes \vec{A}^{(i+R)}}
\end{pmatrix}.
\end{align}
When all the parties follow a LDS, \Cref{eq:TIR_proj} can be written as a sum over $N$ terms that depend only on the strategies of $R$ neighboring parties:
\begin{equation}
\TIRproj(\corrvecNoOneNBody{_{[R]}}_{\boldsymbol{s}}) = \frac{1}{N}\sum_{i=0}^{N-1} \corrvecNoOneNBody{_{[R]}}_{(s_i,\dots,s_{i+R})},
\end{equation}
where $\boldsymbol{s} = (s_0,\ldots,s_{N-1})$ is the $N$-partite strategy vector and $\corrvecNoOneNBody{_{[R]}}_{(s_i,\dots,s_{i+R})}$ is the correlator vector with interaction range $R$ defined as
\begin{align}
    \label{eq:LDS_TIR_correvec}
\begin{pmatrix}
\frac{1}{R+1}\sum_{j=0}^{R}\corrvecNoOne{}_{s_{i+j}} \\
\frac{1}{R}\sum_{j=0}^{R-1} \corrvecNoOne{}_{s_{i+j}} \otimes \corrvecNoOne{}_{s_{i+j+1}} \\
\frac{1}{R-1}\sum_{j=0}^{R-2}\corrvecNoOne{}_{s_{i+j}} \otimes \corrvecNoOne{}_{s_{i+j+2}} \\
\vdots \\
\corrvecNoOne{}_{s_{i}} \otimes \corrvecNoOne{}_{s_{i+R}}
\end{pmatrix}.
\end{align}
As long as $N>2R$, all correlators occurring in \Cref{eq:TIR_proj} are distinct. This means that the projection map $\TIRproj$ is surjective. Since we know that the non-TI local polytope $\localpoly{(N,m)}$ is full-dimensional, we can conclude that the polytope $\localpoly{(N,m)}_{\TIRproj}$ is of full dimension $m+Rm^2$ for $N>2R$. 

Our framework can easily be extended to incorporate higher-order correlators (with bounded interaction range). For this one just needs to add entries to \Cref{eq:TIR_proj} that encode the desired correlators and choose a way to split is as a sum of $N$ terms, each of which only depends on $R+1$ neighboring parties.

\section{Polytopes associated to graphs}
\label{app:graph_theory}

In this section, we fix a finite-oriented graph $\Gamma$ with $n$ vertices. We allow self-loops, but not parallel edges. A closed path in $\Gamma$ is a sequence $c=(i_1,i_2,\ldots,i_\ell)$ of vertices such that all pairs $(i_1,i_2), \ldots, (i_\ell,i_1)$ are edges in $\Gamma$. The length of the closed path is denoted $\ell(c)$. If no vertices are repeated, the path is called a simple cycle. The set of all simple cycles in $\Gamma$ will be denoted $\simpcycleSet{\Gamma}$. Note that this is a finite set. 
To each closed path $c$ we assign a weight matrix 
\[
W(c) = e_{i_1,i_2} + \cdots + e_{i_{\ell},i_1}, 
\]
where $e_{i,j}$ is the matrix with only one nonzero entry, on position $(i,j)$. Explicitly, the $(i,j)$-th entry of $W(c)$ counts how many times $c$ passes through the oriented edge $(i,j)$.

Let us write $\mathcal{W}_{\Gamma} \subset \mathbb{R}^{n \times n}$ for the linear span of all weight matrices of closed paths in $\Gamma$. Explicitly, $\mathcal{W}_{\Gamma}$ consist of all matrices for which the only nonzero entries are in positions corresponding to the edges of $\Gamma$, and moreover for each $i$ the sum of the entries in row $i$ equals the sum of the entries in column $i$. Namely,
\begin{equation}
\label{eq:def_W_gamma}
    \begin{split}
        \mathcal{W}_{\Gamma}=\{M \in \mathbb{R}^{n \times n} \vert &M_{ij}=0~\text{if}~(i,j)\notin \Gamma, \\
    &\sum_j M_{ij}=\sum_l M_{li}~\forall~i \in [n] \}.
    \end{split}
\end{equation} 
The subset of $\mathcal{W}_{\Gamma}$ of matrices whose entries sum to $N$ is denoted $\mathcal{W}_{\Gamma}(N)$. This is an affine-linear space. In particular we have that $W(c) \in \mathcal{W}_{\Gamma}(\ell(c))$. 
 We write $w(c)$ for the normalized weight matrix $W(c)/\ell(c)$. This is an element of $\mathcal{W}_{\Gamma}(1)$.
Then we can define the normalized cycle polytope as in \Cref{eq:definition_pstar}:
\[
p_*(\Gamma) := \operatorname{Conv}\{ w(c) \mid c \text{ simple cycle}\} \subset \mathcal{W}_{\Gamma}(1).
\]
In addition, we can for every $N$ define the normalized closed path polytope $p_N(\Gamma)$ as in \Cref{eq:definition_pN}:
\[
\operatorname{Conv}\{ w(c) \mid c \text{ closed path of length } N\} \subset \mathcal{W}_{\Gamma}(1).
\]
The unnormalized versions of these polytopes will be denoted by $P_N(\Gamma) = N \cdot p_N(\Gamma) \subset \mathcal{W}_{\Gamma}(N)$.
Note a similar construction without normalizing the weight matrices has been made in \cite{balas_cycle_2000}.
As shown in the main text, every TI Bell polytope is obtained as the image of $p_N(\Gamma)$ under a suitable projection map in \Cref{eq:map_weight_matrix}.

\subsection{Relation between the \texorpdfstring{$p_N(\Gamma)$}{P\_N(G)}}
\label{app:relation_periodicityPolytope}

\begin{theorem}
    \label{thm:periodicityPolytope}
    For any $N$, we have an inclusion $p_N(\Gamma) \subseteq p_*(\Gamma)$. If $N$ is divisible by the length of every cycle in $\Gamma$, then $p_N(\Gamma) = p_*(\Gamma)$.
\end{theorem}
Note that from the theorem it immediately follows that
    \begin{align*}
        p_*(\Gamma) &= \bigcup_{N \in \mathbb{N}}{p_N(\Gamma)} \\
        &= \operatorname{Conv}\{ w(c) \mid c \text{ closed path of any length}\},
    \end{align*}
    justifying the name $p_*(\Gamma)$.
\begin{proof}
    To show the first inclusion, we use that any closed path in a graph can be decomposed into simple cycles.
    If a given closed path $c$ (of length $N$) decomposes into (not necessarily distinct) cycles $C_1$, \dots, $C_s$ of lengths $\ell_1$, \dots, $\ell_s$, then we find
    \[
    N w(c) = \ell_1 w(C_1) + \dots + \ell_k w(C_s),
    \]
    from which is follows that $w(c)$ is a convex linear combination of the vectors $w(C_i)$. This shows $p_N(\Gamma) \subseteq p_*(\Gamma)$.
    
    Now assume that $N$ is divisible by the length of every cycle in $\Gamma$. Then if we fix a cycle $C$ of length $\ell$,
    we have $w(c)=w(\frac{N}{\ell}\cdot C) \in p_N(\Gamma)$, where $\frac{N}{\ell}\cdot C$ is the (length $N$) path obtained by repeating $\frac{N}{\ell}$ times the cycle $C$. This shows that $p_N(\Gamma) = p_*(\Gamma)$. 
\end{proof}

In addition, the sequence of polytopes $p_N(\Gamma)$ ``converges" to the polytope $p_*(\Gamma)$. There are several ways of formalizing the notion of convergence. For instance: the vertices of $p_N(\Gamma)$ will lie closer and closer to the vertices of $p_*(\Gamma)$ when $N \to \infty$. Or: for every inequality, the corresponding bound on $p_N(\Gamma)$  converges to the one on $p_*(\Gamma)$. But note that the number of vertices of $p_N(\Gamma)$ is not a convergent sequence; we will soon see that it is in fact eventually periodic. 

\subsection{Computing the vertices of \texorpdfstring{$p_N(\Gamma)$}{P\_N(G)}}
The vertices of $p_*(\Gamma)$ are easy to describe: they are the weight matrices $w(c)$, where $c$ runs over all cycles in $\Gamma$. This is true because the weight matrix of a simple cycle cannot be written as a convex combination of weight matrices of other cycles.

For a finite $N$ which is not divisible by the length of every cycle in $\Gamma$, $p_N(\Gamma)$ might have many more vertices. Intuitively, going from $p_*(\Gamma)$ to $p_N(\Gamma)$ means that we replace vertex $w(c)$ with a cluster of vertices lying close to $w(c)$. However, it turns out that the behavior near $w(c)$ only depends on the value of $N$ modulo $\ell(c)$. 

\begin{theorem} \label{thm:exactVertices}
    Fix a graph $\Gamma$. Then there exists for every cycle $c$ and every $r=0,\ldots,\ell(c)-1$ a finite list $\err{\Gamma}{c}{r}$ of matrices in $\mathcal{W}_{\Gamma}(0)$, such that for any sufficiently large $N$ the vertices of $p_N(\Gamma)$ are of the form
    \begin{equation}\label{eq:exactVertices1}
    w(c)+\frac{e}{N} \text{ with } e\in \err{\Gamma}{c}{N \operatorname{mod} \ell(c)}.
    \end{equation}
\end{theorem}
We will later be more precise about what ``sufficiently large" means, see \Cref{rmk:sufficientlyLarge}. 
This implies that the number of vertices of $p_N(\Gamma)$ is at most
\[
\sum_{c \in \simpcycleSet{\Gamma}}{\max_{r=0,\ldots,\ell(c)-1}}{|\err{\Gamma}{c}{r}|},
\]
which is independent of $N$.

In the rest of this section, we explain why this theorem is true, and how to compute these ``lists of error terms" $\err{\Gamma}{c}{r}$. We are going to need some more definitions from polyhedral geometry. We hope they will all be clear from the context, but the reader who desires more background can for instance consult \cite{ziegler_lectures_1995}. 
In the next section, we show the results for the complete graph on 4 vertices.

\subsubsection{Construction of \texorpdfstring{$\err{\Gamma}{c}{r}$}{err\_\{G\}(c,r)}}

Before we can define the error terms $\err{\Gamma}{c}{r}$, we need to introduce some properties that our weight matrices can satisfy. To a matrix $M \in \mathcal{W}_{\Gamma}$ we associate a subgraph $\Gamma_M \subset \Gamma$, whose edges correspond to the nonzero entries in $M$. The matrix $M$ is irreducible if and only if $\Gamma_M$ is strongly connected. We will need two technical variants of this definition. 
\begin{definition} \label{def:irreducibleMatrix}
    A matrix $M \in \mathcal{W}_{\Gamma}$ is \emph{weakly irreducible} if $\Gamma_M$ is connected after removing isolated vertices.
\end{definition}
If the entries of a matrix $M \in \mathcal{W}_{\Gamma}$ are nonnegative integers, being weakly irreducible is equivalent to being the weight matrix of a closed path. On the other hand, if we sum the weight matrices of 2 simple cycles that don't share a vertex, we obtain a matrix that is not weakly irreducible.
\begin{definition} \label{def:c-irreducibleMatrix}
    Let $M \in \mathcal{W}_{\Gamma}$ and $c \in \simpcycleSet{\Gamma}$ a simple cycle.     
    Then $M$ is \emph{$c$-irreducible} if the graph obtained by adding the edges of $c$ to $\Gamma_M$, is connected. The vertices of this graph form the \emph{$c$-support} of $M$.   
\end{definition}
This will be a relevant condition on our error terms: if $e$ is $c$-irreducible, then the matrix from \Cref{eq:exactVertices1} will be weakly irreducible, i.e.\ it is the weight matrix of a closed path. 
Now we are ready to mathematically define the error terms $\err{\Gamma}{c}{r}$. 
\begin{definition}\label{def:errorTerms}
    For every simple cycle $c$ in $\Gamma$ and every $r=0,1,\dots,\ell(c)-1$ we make the following construction:
    \begin{itemize}
        \item Consider the cone $C_{c,r}$ with vertex $r\cdot w(c)$ and rays in the directions $w(c')-w(c)$, where $c'$ runs over all other simple cycles of $\Gamma$. This cone lives in $\mathcal{W}_\Gamma(r)$. 
        \item Inside $C_{c,r}$, we take all matrices $M$ that satisfy the two conditions:
        \begin{itemize}
            \item The entries of $M$ are integers, i.e.\ $M$ is a lattice points.
            \item $M$ is $c$-irreducible.
        \end{itemize}
        The convex hull of all such matrices is called $\tilde{C}_{c,r}$. 
        \item Define $\errtilde{\Gamma}{c}{r} \subset \mathcal{W}_\Gamma(r)$ to be the set of all vertices in $\tilde{C}_{c,r}$, (i.e.\ all points in $\tilde{C}_{c,r}$ that are not a convex combination of other points). 
        \item Finally, define 
        \begin{equation} \label{eq:errorterms}
        \err{\Gamma}{c}{r} := \{\tilde{e} - r \cdot w(c) \mid \tilde{e} \in \errtilde{\Gamma}{c}{r}\}).
        \end{equation}
    \end{itemize}
\end{definition}
Now that we have defined $\err{\Gamma}{c}{r}$, there are three things to do:
\begin{enumerate}
    \item Show that is the set we wanted, in other words, prove \Cref{thm:exactVertices}.
    \item Show that $\err{\Gamma}{c}{r}$ is actually a finite set.
    \item Find an algorithm to compute it in practice.
\end{enumerate}
In the remainder of this subsection we deal with the first point; the other two will be handled in the next subsections. 
\begin{proof}[{Proof of \Cref{thm:exactVertices}}]
    We need to show that $P_N(\Gamma)$ is the convex hull of all lattice points 
    \begin{equation} \label{eq:correctionTerm}
    a \cdot W(c) + v \text{ with } v \in \errtilde{\Gamma}{c}{r}
    \end{equation}
    where $N=a\cdot \ell(c) + r$.

    Note that the matrices \Cref{eq:correctionTerm} have nonnegative entries: by construction the only negative entries $v \in \tilde{L}_{c,r}$ are in the positions corresponding to $c$, so if we choose $N$ (and hence $a$) large enough they become nonnegative. We can even assume that the entries on the positions corresponding to $c$ are positive. Then (by definition of $c$-irreducibility) all matrices \Cref{eq:correctionTerm} are weakly irreducible and have nonnegative integer entries. Hence they are all weight matrices of closed paths of length $N$, i.e.\ they lie in $P_N(\Gamma)$. 

    We now show that every vertex of $P_N(\Gamma)$ is actually of the form \Cref{eq:correctionTerm}. 
    So let $p$ be a vertex of $P_N(\Gamma)$. Let $F$ denote the smallest face of $N\cdot p_*(\Gamma)$ that contains $p$. Let $C_F$ denote the cycles corresponding to the vertices of $F$. Note that $p$ 
    is weakly $c$-irreducible for $c \in C_F$. This is true since $p$ itself is weakly irreducible, and the entries corresponding to $c$ are already nonnegative.

    We can find a linear function $\xi$ such that the unique minimum of $\xi$ on  $P_N(\Gamma)$ is achieved at $p$, and the unique minimum of $\xi$ on $N\cdot p_*(\Gamma)$ is achieved at $w(c)$ for some $c$ in $C_F$. Indeed: start with any $\xi$ that has its unique minimum on  $P_N(\Gamma)$ at $p$ (this exists because $p$ is a vertex), perturb the coefficients such that all numbers $\xi(w(c))$ are distinct, then add to it a linear function that is $0$ on $F$ and takes large positive values on the $w(c)$ with $c \notin C_F$. 

    Let $c_0$ be the cycle for which $\xi(w(c)))$ is minimal. Then $\xi$ is bounded from below on the cone $C_{c_0,N_{c_0}}$. So we can take the $c_0$-stably irreducible lattice point in $C_{c_0,N_{c_0}}$ for which $\xi$ is minimal. Call this lattice point $v$. Note that $v$ is a vertex of $\tilde{C}_{c,N_{c_0}}$. 

    Now we can translate the cone $C_{c_0,N_{c_0}}$ to place its vertex at $Nw(c_0)$: the resulting cone is $ \lfloor\frac{N}{\ell(c_0)}\rfloor W(c_0)+C_{c_0,N_{c_0}}$, and it contains the polytope $N\cdot p_*(\Gamma)$. The point in this translated cone at which $\xi$ is minimal is $\lfloor\frac{N}{\ell(c)}\rfloor\cdot W(c) + v$, but this point lies is $N\cdot p_*(\Gamma)$ and hence it must be equal to the point $p$ we started with. This shows that $p$ has the desired form \Cref{eq:correctionTerm}. 
\end{proof}

\begin{remark} \label{rmk:sufficientlyLarge}
Once we have computed $\err{\Gamma}{c}{r}$, we can compute how large $N$ needs to be for Theorem \ref{thm:exactVertices} to be true: if all matrices \Cref{eq:correctionTerm} have positive entries on the positions corresponding to $c$.
\end{remark}

\subsubsection{Proof of finiteness and inefficient computation}
\label{proof_of_finiteness_and_inefficient_computation}
In what follows, we will work with the set $\errtilde{\Gamma}{c}{r}$ instead of $\err{\Gamma}{c}{r}$, since the former consists of matrices with integer entries. 
We will now show that $\errtilde{\Gamma}{c}{r}$ is a finite set; which will also give us an inefficient algorithm to compute it. Recalling \Cref{def:errorTerms}, we start by considering the cone $C_{c,r}$. 
The cone $C_{c,r}$ consists of all points
\begin{equation}\label{eq:pointInCone}
\lambda_{c} W(c) + \sum_{c' \in \simpcycleSet{\Gamma} \setminus c}{\lambda_{c'}W(c')} \text{ with } \sum_{c' \in \simpcycleSet{\Gamma}}{\lambda_{c'}\ell(c')}=r
\end{equation}
where $\lambda_{c}$ can be be negative but the other $\lambda_{c'}$ have to be nonnegative. Equivalently: $C_{c,r}$ is the set of all points in $\mathcal{W}_{\Gamma}(r)$ for which the coordinates not corresponding to $c$ are nonnegative.

Now we need to take the set of all $c$-irreducible lattice points in $C_{c,r}$. Clearly whenever all $\lambda_{c'}$ are integers, \Cref{eq:pointInCone} is a lattice point. With some more work, one sees that actually every lattice point in $C_{c,r}$ can be written as \Cref{eq:pointInCone} with all coefficients integers (the idea is to add a large multiple of $W(C_0)$ so all entries become positive, and then use decomposition of a cycle into simple cycles).
The $c$-irreducibility only depends on which of the $\lambda_{c'}$ (with $c' \neq c$) are zero and which are not.

\begin{definition}
We make the following constructions
\begin{itemize}
    \item $\mathcal{E}_{c,r}$ is the set of $c$-irreducible lattice points in $C_{c,r}$. 
    This set is infinite.
    \item $\mathcal{E}'_{c,r}$ is the subset of points $p$ in $\mathcal{E}_{c,r}$ for which $p+\lcm{\ell(c),\ell(c')}(w(c)-w(c'))$ is no longer in $\mathcal{E}_{c,r}$ for all $c' \neq c$.
    This is a finite set.
\end{itemize}
\end{definition}
\begin{remark} \label{rem:finite_set_statement}
If one of the coefficients $\lambda_{c'}$ (with $c' \neq c$) is larger than $\lcm{\ell(c),\ell(c')}/\ell(c')+1$, the sum in \Cref{eq:pointInCone} will not lie in $\mathcal{E}'_{c,r}$.
\end{remark}

Furthermore, one sees that the set $\tilde{C}_{c,r} = \Conv(\mathcal{E}_{c,r})$ equals the polyhedron generated by the points $\mathcal{E}'_{c,r}$ and the ray directions $w(c')-w(c)$. So we have reduced it to a finite computation:
\begin{itemize}
    \item Compute the finite set $\mathcal{E}'_{c,r}$.
    \item Compute the polyhedron generated by the points $\mathcal{E}'_{c,r}$ and the ray directions $w(c')-w(c)$.
    \item Compute the vertices of this polyhedron. This is $\errtilde{\Gamma}{c}{r}$.
\end{itemize}
In particular, this shows that $\errtilde{\Gamma}{c}{r}$ is actually finite.

The final two points can be done using standard polyhedral software such as \texttt{normaliz}~\cite{bruns_normaliz_2017}. The first one can in principle be done by running over all vectors \Cref{eq:pointInCone} where $\lambda_{c'} \leq \lcm{\ell(c),\ell(c')}/\ell(c')+1$ for all $c' \neq c$, and $\lambda_{c}$ is uniquely determined by the equality on the right. But this is very slow in practice.

\subsubsection{Making the computation more efficient}

A lattice point in $C_{c,r}$ can be decomposed as \Cref{eq:pointInCone} in many different ways. 
To get rid of this redundancy, we triangulate the cone $C_{c,r}$. This means we compute a set $\Delta$ of subsets of $\simpcycleSet{\Gamma}\setminus {c}$ (referred to as the \emph{simplices} of the triangulation), such that 
\begin{itemize}
    \item $\Delta$ is closed under taking subsets,
    \item for every $S \in \Delta$, the vectors $\{w(c')-w(c)\mid c' \in S\}$ are linearly independent,
    \item every point in $C_{c,r}$ can be uniquely written in the form
\begin{align}
& r w(c) + \sum_{c' \in S}{\mu_{c'}(w(c')-w(c))} \\
=& \lambda_{c} W(c) + \sum_{c' \in S}{\lambda_{c'}W(c')} \label{eq:triangulatedCone}
\end{align} for some (uniquely defined) $S \in \Delta$. Here $\lambda_{c}$ is arbitrary and $\lambda_{c'}$ are strictly positive, and $\mu_{c'}=\ell(c)\lambda_{c'}$. Such a point is a lattice point if and only if all the $\lambda$'s are integers. 
\end{itemize}
To put it geometrically: every $S \in \Delta$ defines a simplicial cone $\operatorname{Cone}_S$, defined by the rays $\{w(c')-w(c)\mid c' \in S\}$, and every point in $C_{c,r}$ lies is in the interior of exactly one of these cones.

This triangulation can also be computed using \texttt{normaliz}. 
Now the idea is to go through all simplices in $\Delta$, and see if they could possibly contain vertices of $\tilde{C}_{c,r}$. 
For every simplex $S$, we compute the corresponding matrix 
\[
M_S = \sum_{c' \in S} W(c').
\]
The following three claims allow us to discard many simplices:

\begin{claim} \label{claim:connected_simplex}
    If $M_S$ is not a $c$-irreducible matrix, the points in the interior of $\operatorname{Cone}_S$ will also not be $c$-irreducible. In particular, they cannot be vertices of $\tilde{C}_{c,r}$. \qed
\end{claim} 

\begin{claim} \label{claim:non_extremal_simplex_1}
    Suppose $M_S':=M_S-\lcm{\ell(c),\ell(c')}w(c')$ has nonnegative entries outside of $c$, and is $c$-irreducible with the same $c$-support as $M_S$, for some $c' \in S$. Then for every lattice point $p$ in the interior of $\operatorname{Cone}_S$,
    $p$ is not a vertex of $\tilde{C}_{c,r}$.

    If moreover this $c$-support is all of $[n]$, then for every $S' \in \Delta$ that contains $S$, none of the points in the interior of $\operatorname{Cone}_{S'}$ are vertices of $\tilde{C}_{c,r}$. 
\end{claim}

\begin{proof}
It suffices to show that $p':=p+\lcm{\ell(c),\ell(c')}(w(c)-w(c'))$ is still in $\mathcal{E}_{c,r}$. 
By construction, $p'$ is still a lattice point in our space $\mathcal{W}_\Gamma(r)$. Moreover, since outside of $c$ the entries of $p$ are larger than those of $M_S$, we see that $p'$ has nonnegative entries outside of $c$. So $p'$ is a lattice point inside $C_{c,r}$, we still need to show it is $c$-irreducible. This follows from the fact that $M_S'$ is already $c$-irreducible and that $p'$ can be obtained from $M_S'$ by adding multiples of $W(c')$, $c' \in S$.

The proof of the ``moreover" part is analogous and left to the reader.
\end{proof}

\begin{claim} \label{claim:non_extremal_simplex_2}
    Let $c_1,c_2 \in S$ of the same length, and
    suppose both $M_S+w(c_1)-w(c_2)$ and $M_S+w(c_2)-w(c_1)$ have nonnegative entries outside of $c$, and are $c$-irreducible with the same $c$-support as $M_S$. Then for every lattice point $p$ in the interior of $\operatorname{Cone}_S$, 
    $p$ is not a vertex of $\tilde{C}_{c,r}$.

   If moreover this $c$-support is all of $[n]$, then for every $S' \in \Delta$ that contains $S$, none of the points in the interior of $\operatorname{Cone}_{S'}$ are vertices of $\tilde{C}_{c,r}$.
\end{claim}
\begin{proof}
    We have $p \in \Conv\{p+w(c_1)-w(c_2),p-w(c_1)+w(c_2)\}$, the only thing to verify is that both these matrices are $c$-irreducible lattice points in $C_{c,r}$, which we leave to the reader. 
\end{proof}

Finally, for all simplices $S$ we didn't discard, we compute all vectors
\begin{equation}\label{eq:pointInCone2}
\lambda_{c} W(c) + \sum_{c' \in S}{\lambda_{c'}W(c)} \text{ with } \sum_{c' \in S}{\lambda_{c'}\ell(c)}=r,
\end{equation}
where $\lambda_{c'} \leq \lcm{\ell(c),\ell(c')}/\ell(c')+1$ for all $c' \in S$, and $\lambda_{c}$ is uniquely determined by the equality on the right. For each of them, we manually check if they are in $\mathcal{E}'_{c,r}$. 

The algorithm to find the error terms is summarized in \Cref{algo}.  
For ease of notation, we denote $\ell_c$ for $\ell(c)$, $w_c$ for $w(c)$ and $W_c$ for $W(c)$. 
The input of the algorithm is an oriented graph $\Gamma$, a simple cycle $c_0$ and an integer $r \in \{0, \dots, \ell(c_0)-1\}$. 
The algorithm computes the list of error terms $\errtilde{\Gamma}{c_0}{r}$.

\begin{figure*}
\hrule
\textbf{Algorithm}
\hrule 
\begin{algorithmic}
\Require $\bullet$ a finite oriented graph $\Gamma$ with self-loops, but no parallel edges \\
         $\bullet$ a simple cycle $c_0 \in \simpcycleSet{\Gamma}$ \\
         $\bullet$ an interger $r \in \{0,\dots,\ell_{c_0}-1\}$
\Statex
\LComment{Construct the cone $C_{c_0,r}$ with vertex $r\cdot w_{c_0}$ and rays $\{(w_c-w_{c_0}) \mid c \in \simpcycleSet{\Gamma}\}$}
\State $C_{c_0,r} \gets \Call{cone}{r\cdot w_{c_0},\{w_c-w_{c_0}\}_{c \in \simpcycleSet{\Gamma}}}$
\Statex
\LComment{Compute the triangulation of the cone}
\State $\Delta \gets \Call{triangulate}{C_{c_0,r}}$
\Statex
\LComment{Discard simplices according to claims \ref{claim:connected_simplex},\ref{claim:non_extremal_simplex_1} and \ref{claim:non_extremal_simplex_2}}
\State $\Delta \gets \Call{apply\_claims}{c_0,\Delta}$
\Statex
\LComment{For all the remaining simplices, find the points that are in $\mathcal{E}'_{c,r}$. Start with empty set}
\State $\mathcal{E}'_{c_0,r} \gets \emptyset$
\For{$S \in \Delta$}
    \LComment{According to the remark \ref{rem:finite_set_statement}, $\lambda_c=1,\dots,\lcm{\ell_{c_0},\ell_c}/\ell_c$ }
    \For{$(\lambda_c)_{c \in S}$ with $\lambda_c=1,\dots,\frac{\lcm{\ell_{c_0},\ell_c}}{\ell_c}+1$}
        \LComment{Since we want points in the cone, we need $\lambda_{c_0}\ell_{c_0} + \sum_{c \in S} \lambda_c \ell_c = r$}
        \If{$\sum_{c \in S} \lambda_c \ell_c \mod \ell_{c_0} = r$} 
            \LComment{Since we want lattice points, $\lambda_{c_0}$ is the integer satisfying the cone constraint}
            \State $\lambda_{c_0} \gets - \lfloor \frac{1}{\ell_{c_0}}\sum_{c \in S} \lambda_c \ell_c \rfloor$
            \LComment{Lattice point in the cone $C_{c_0,r}$}
            \State $\epsilon \gets \lambda_{c_0} W_{c_0} + \sum_{c \in S} \lambda_{c} W_c$
            \If{$\epsilon$ is $c_0$-irreducible}
                \LComment{$\epsilon \in \mathcal{E}_{c_0,r}$ because it is a $c_0$-irreducible lattice point in the cone. Now we check if $\epsilon$ is also in $\mathcal{E}_{c_0,r}'$}
                \State $\epsilon$\_in\_$\mathcal{E}_{c_0,r}'$ $\gets$ True 
                \For{$c \in S$}
                    \State $\epsilon' \gets \epsilon + \lcm{\ell_{c_0},\ell_{c}}(w_{c_0}-w_{c})$
                    \If{$\epsilon'$ is $c_0$-irreducible and has no negative entries outside $c_0$}
                    \LComment{$\epsilon' \in \mathcal{E}_{c_0,r}$ because it is a $c_0$-irreducible lattice point in the cone.
                    Thus $\epsilon \notin \mathcal{E}_{c_0,r}'$.}
                        \State $\epsilon$\_in\_$\mathcal{E}_{c_0,r}'$ $\gets$ False
                        \Break
                    \EndIf
                \EndFor
                \If{$\epsilon$\_in\_$\mathcal{E}_{c_0,r}'$}
                    \State $\mathcal{E}'_{c_0,r} \gets \mathcal{E}'_{c_0,r} \cup \{\epsilon\}$
                \EndIf
            \EndIf
        \EndIf
    \EndFor
\EndFor
\Statex
\LComment{Construct the polyhedron $\tilde{C}_{c_0,r}$ with the points $\mathcal{E}'_{c_0,r}$ and rays $\{(w_c-w_{c_0}) \mid c \in \simpcycleSet{\Gamma}\}$}
\State $\tilde{C}_{c_0,r} \gets \Call{cone}{\mathcal{E}'_{c_0,r} \;,\{w_c-w_{c_0}\}_{c \in \simpcycleSet{\Gamma}}}$
\Statex
\LComment{The error list correspond to the vertices of $\tilde{C}_{c_0,r}$}
\State $\errtilde{\Gamma}{c_0}{r} \gets \Call{vertices}{\tilde{C}_{c_0,r}}$
\end{algorithmic}
\hrule
\caption{Algorithm to compute the list of error term $\errtilde{\Gamma}{c}{r}$ as defined in \Cref{def:errorTerms}.
These error terms allow us to compute the vertices of $P_N(\Gamma)$ according to \Cref{eq:correctionTerm}. Note that this algorithm is independent of $N$.}
\label{algo}
\end{figure*}

\subsection{The case of \texorpdfstring{$K_4$}{K\_4}: results and example}

We used the algorithm from the previous section to compute the error lists $\errtilde{K_4}{c}{r}$ in the case of the complete graph of 4 nodes (with self-loops). 
\Cref{tab:errorTerms1} records the length of these error lists, i.e.\ the number of vertices of $p_N(K_4)$ that lie close to $w(c)$. 
For symmetry reasons, these numbers depend on the length of the cycle $c$ (and on $r$).  
For instance, for $N=19$, we have remainder $3$ after division by $4$, so the polytope $p_{19}(K_4)$ has $24$ vertices close to each of the six points $w(c)$ with $\ell(c)=4$. 
Repeating for the other cycle lengths, we find that $p_{19}(K_4)$ has in total 
\[
4\cdot1 + 6\cdot22 + 8\cdot18 + 6\cdot24 =424
\]
vertices. 
We can similarly compute the number of vertices of $p_{N}(K_4)$ for any $N$, the results only depend on the value $N$ modulo $12$ and are presented in the second column \Cref{tab:errorTerms3}. 
As pointed out earlier, these results only hold for $N$ sufficiently large. 
Explicit computations indicate that in this context \enquote{sufficiently large} means $N \geq 18$.
\begin{table}[h]
    \centering
    \begin{tabular}{c | c || c | c | c | c }
    cycle length & $\#$ cycles & $r=0$ & $r=1$ & $r=2$ & $r=3$ \\
    \hline 
    1 & 4 & 1 & & & \\
    2 & 6 & 1 & 22 & & \\
    3 & 8 & 1 & 18 & 18 & \\
    4 & 6 & 1 & 24 & 12 & 24 \\
    \end{tabular}
    \caption{The sizes of $\errtilde{K_4}{c}{r}$.}
    \label{tab:errorTerms1}
\end{table}

The local polytope $\localpoly{(N,2)}_{\TINNproj}$ is the projection of $p_{N}(K_4)$ under the map $\Phi$ from \Cref{eq:weight matrix_to_corrvec}.
Recall from \Cref{subsec:graphTINN} that $\localpoly{(*,2)}_{\TINNproj}$ has $20$ vertices, corresponding to all of the cycles of length $1,2,3$ in $K_4$ and $2$ of the $6$ cycles of length $4$. 
Similar to before, each vertex of $\localpoly{(N,2)}_{\TINNproj}$ lies close to one of these 20 vertices; the ``error terms" can be computed by applying the projection $\Phi$ to the cones $\tilde{C}_{c,r}$ we computed earlier, and taking the vertices of the resulting polyhedra. 
It turns out these numbers again only depend on the cycle length (and on $r$), they can be found in \Cref{tab:errorTerms2}. 
As before, we can use this to compute the number of vertices of $\localpoly{(N,2)}_{\TINNproj}$ for any $N>18$, see the third column of \Cref{tab:errorTerms3}. 
We see that $\localpoly{(N,2)}_{\TINNproj}$ can have at most $200$ vertices:  

\begin{table}[h]
    \centering
    \begin{tabular}{c | c || c | c | c | c }
    cycle length & $\#$ cycles & r=0 & r=1 & r=2 & r=3 \\
    \hline 
    1 & 4 & 1 & & & \\
    2 & 6 & 1 & 14 & & \\
    3 & 8 & 1 & 8 & 8 & \\
    4 & 2 & 1 & 24 & 12 & 24 \\
    \end{tabular}
    \caption{Number of vertices of $\localpoly{(N,2)}_{\TINNproj}$ close to each vertex of $\localpoly{(*,2)}_{\TINNproj}$.}
    \label{tab:errorTerms2}
\end{table}

\begin{table}[h]
    \centering
    \begin{tabular}{c || c | c }
    $(N \operatorname{mod} 12)$ & $p_N(K_4)$  & $\localpoly{(N,2)}_{\TINNproj}$ \\
    \hline
    0 & 24 & 20 \\
    1 & 424 & 200 \\
    2 & 226 & 98 \\
    3 & 288 & 144 \\
    4 & 160 & 76 \\
    5 & 424 & 200 \\
    6 & 90 & 42 \\
    7 & 424 & 200 \\
    8 & 160 & 76 \\
    9 & 288 & 144 \\
    10 & 226 & 98 \\
    11 & 424 & 200 \\
    \end{tabular}
    \caption{Number of vertices of the closed path polytope $p_N(K_4)$ and the TINN Bell polytope $\localpoly{(N,2)}_{\TINNproj}$ for each value of $N>18$.}
    \label{tab:errorTerms3}
\end{table}

\section{Symmetry classes for TI-Bell inequalities} \label{app:symmetry_class}

In general, we consider that two inequalities are equivalent if they are equal under the permutation of parties, input and outcome labels.
For one party, there are $2m-1$ generators of permutations.
The first $m-1$ generators are the input labels permutations $\alpha_{x} \mapsto \alpha_{\sigma(x)}$, where $\sigma$ is a generator of the symmetric group $\text{Sym}(m)$. The $m$ remaining generators are given by the outcome labels permutations for input $x'$: $\alpha_{x} \mapsto -\alpha_x$ if $x=x'$.
By composing these $2m-1$ generators, one gets all possible input and outcome label permutations for a single party.
Since we are working with TI inequalities, the same input and outcome permutations must be applied to all the parties. 
Finally, in TI inequalities, there is also one possible party permutation.
It corresponds to the flipped permutation: $ABCD \rightarrow DCBA$. 
Under this permutation, the neighbors of a given party remain unchanged.
This holds for any interaction range $R$.
We have thus in total $2m$ generators for TI-$R$ inequalities.

The classical bound $\beta$ remains unchanged under these permutations. The symmetry group under consideration has $2^{m+1} \cdot m!$ elements: $m!$ input permutations, $2^m$ outcome permutations and $2$ party permutation. 
There are thus at most $2^{m+1} \cdot m!$ inequalities per symmetry class. 

For example in the case $R=m=2$, there are four permutation generators: parties are flipped $\alpha_{x,y}^{i,j} \rightarrow \alpha_{y,x}^{i,j}$, inputs are flipped, outcomes are flipped when input is $0$ and outcomes are flipped when input is $1$:
\begin{equation}
\begin{aligned} 
&\text{input flip} & \quad & \text{outcomes flip for input $x' \in \{0,1\}$} \\ 
&\alpha^{i}_{x} \mapsto \alpha^{i}_{x+1} & \quad &  \alpha^{i}_{x} \mapsto (-1)^{\delta_{x,x'}}\alpha^{i}_{x} \\
&\alpha^{i,i+1}_{x,y} \mapsto \alpha^{i,i+1}_{x+1,y+1} & \quad & \alpha^{i,i+1}_{x,y} \mapsto (-1)^{\delta_{x,x'}+\delta_{y,x'}}\alpha^{i,i+1}_{x,y}\\
&\alpha^{i,i+2}_{x,z} \mapsto \alpha^{i,i+2}_{x+1,z+1} & \quad & \alpha^{i,i+2}_{x,z} \mapsto (-1)^{\delta_{x,x'}+\delta_{z,x'}}\alpha^{i,i+2}_{x,z} \\
&\beta \mapsto \beta & \quad & \beta \mapsto \beta
\end{aligned}
\end{equation}
where all the input indices are taken modulo $m$. 
In the $(N,2)$ scenario with $R=2$ and $N$ divisible enough, the projected polytope has $32372$ facets and there are $2102$ different symmetric classes ~\cite{wang_entanglement_2017}.

\section{Renormalization computation} \label{Appendix:renormComp}
Suppose we have a tropical $n \times n$ matrix $F'(\boldsymbol{a})$ whose entries are linear functions of the variables $\boldsymbol{a}$. We can assume that this parametrization is injective, i.e.\ $F'(\boldsymbol{a})=0 \implies \boldsymbol{a}=0$. We want to find the solutions to 
\begin{equation} \label{eq:renormalizationAppendix}
    F'(\boldsymbol{a})^{\odot 2} = F'(\boldsymbol{a}),
\end{equation}
i.e.\ we want to find the values of $\boldsymbol{a}$ such that for all $i,j \in [n]$:
\begin{equation}
    \min_k{F'_{i,k}(\boldsymbol{a})+F'_{k,j}(\boldsymbol{a})}=F'_{i,j}(\boldsymbol{a}).
\end{equation}
To this end, we consider the polyhedron $\mathcal{P}_{F'}$ defined by the $n^3$ inequalities

\begin{equation} \label{eq:coneInequalities}
    L_{i,j,k}(\boldsymbol{a}) := F'_{i,k}(\boldsymbol{a})+F'_{k,j}(\boldsymbol{a}) - F'_{i,j}(\boldsymbol{a}) \geq 0. 
\end{equation}
Since $\mathcal{P}_{F'}$ is defined by \emph{homogeneous} linear inequalities it is a so-called polyhedral cone. A version of the Minkowski-Weyl theorem (to be precise, we need the additional requirement that the cone is \emph{pointed}, i.e.\ that it doesn't contain a linear space, this follows from the injectivity assumption) states that such a cone is the conical hull of finitely many vectors $\boldsymbol{a}_1,\dots,\boldsymbol{a}_c$:
\begin{equation}
    \mathcal{P}_{F'} =\{\sum_\ell{\lambda_\ell\boldsymbol{a}_\ell} \mid \lambda_\ell \in \mathbb{R}_{\geq 0}\}.
\end{equation}
The half-lines $\mathbb{R}_{\geq 0}\boldsymbol{a}_{\ell}$ are known as the \emph{rays} of the cone, the $\boldsymbol{a}_{\ell}$ themselves are the \emph{ray generators}. Our polyhedral cone can be decomposed into faces, where every face is the conical hull of a subset of the rays. Now the key observations to make are
\begin{enumerate}
    \item The solutions to \Cref{eq:renormalizationAppendix} are exactly the points in $\mathcal{P}_{F'}$ for which sufficiently many of the inequalities in \Cref{eq:coneInequalities} become equalities: if we group the inequalities in $n^2$ groups, we need that in each group one of them is an equality.
    \item The collection of inequalities $L_{i,j,k}(\boldsymbol{a})$ that are achieved only depend on which face lies in; and the faces where the most inequalities are achieved are the rays.
\end{enumerate}
This means that we can find solve \Cref{eq:renormalizationAppendix} by performing the following three steps:
\begin{enumerate}
    \item Compute the ray generators of the polyhedron 
    $\mathcal{P}_{F'}$ defined by the inequalities in \Cref{eq:coneInequalities}.
    \item For each ray generator, make a list of which inequalities are achieved.
    \item Now it is easy to list the faces of $\mathcal{P}_{F'}$ where enough of the inequalities are achieved; the solution set to \Cref{eq:renormalizationAppendix} is the union of these faces.
\end{enumerate}
The computational bottleneck appears to be the first step. We wonder if it is possible to speed up the algorithm by avoiding having to compute all the rays of $\mathcal{P}_{F'}$.

\section{Computing the violation of Bell inequalities numerically} \label{app:DMRG}
In this section, the numerical methods used to compute the violation of the TI-$2$ inequalities in the $(N,2,2)$ scenario are presented.

Finding the violation can be stated as a minimization problem:
\begin{equation}
\min \operatorname{Tr}[W\rho],
\end{equation}
where $\rho$ is a quantum state and $W$ is the quantum operator associated to the TI Bell inequality:
\begin{equation}
W = \frac{1}{N}\sum_{i=0}^{N-1} \left(\sum_{x=0}^{1} \alpha_x A^{(i)}_x + \sum_{x,y=0}^1 \alpha_{x,y} A^{(i)}_x A_y^{(i+1)}\right) , 
\end{equation}
with $\{A_{x}^{(i)}\}_{i,x}$ being the measurement operators (POVM).
Following~\cite{wang_entanglement_2017}, the measurement operators are assumed to have the form: $A^{(i)}_0 =  \dots \mathbb{I} \otimes M(0,0) \otimes \mathbb{I} \dots$ and $A^{(i)}_1 = \dots \mathbb{I} \otimes M(\theta_1,\theta_2) \otimes  \mathbb{I} \dots$ where the operator
\begin{equation}
M(\theta_1,\theta_2) := 
\begin{pmatrix}
\cos \theta_1 & \sin \theta_1 & 0 & 0 \\
\sin \theta_1 & -\cos \theta_1 & 0 & 0 \\
0 & 0 & \cos \theta_2 & \sin \theta_2 \\
0 & 0 & \sin \theta_2 & -\cos \theta_2 \\
\end{pmatrix}
\end{equation}
acts on the $i$-th site.

The method used here is a seesaw minimization. 
First, the optimization is performed over the state $\rho$ using the DMRG method~\cite{white_density_1992,schollwock_density-matrix_2005,schollwock_density-matrix_2011}. 
Next, the optimization is performed over the measurements by minimizing with respect to the parameters $\theta_1$ and $\theta_2$ using the BFGS method~\cite{broyden_convergence_1970,fletcher_new_1970,goldfarb_family_1970,shanno_conditioning_1970}.
This procedure is repeated until convergence.
Note that this method does not give convergence guarantees, however, our empirical tests show that convergence is reached reliably.

The implementation of the DMRG method uses the python library TenPy~\cite{hauschild_efficient_2018} with the following parameters: bond dimension is $50$, maximal number of sweeps is $100$, minimum cut-off of the SVD is $10^{-10}$, convergence is assumed when the energy variation is smaller than $10^{-5}$ and for the remaining parameters the default values are used.
The implementation of the BFGS method uses the python library SciPy~\cite{virtanen_scipy_2020} with the \texttt{l-bfgs-b} minimization method with the default values.
Since the local-minimum to which the BFGS algorithm converges is sensitive to the initial guess of the parameters $(\theta_1,\theta_2)$, the optimizations were repeated five times, each time with different initial guesses randomly chosen.
When the difference between the local bound and the minimization result is at least $10^{-2}$, it is assumed that the inequality is violated.
Details on implementation are given \href{https://gitlab.com/elo_val/ti-bell-inequalities}{online}~\footnotemark[0].

\footnotetext[0]{The code supporting this work is available at \href{https://gitlab.com/elo_val/ti-bell-inequalities}{\url{gitlab.com/elo_val/ti-bell-inequalities}}}.

\bibliography{references}

\end{document}